\documentclass[10pt, conference, letterpaper]{IEEEtran}
\def\BibTeX{{\rm B\kern-.05em{\sc i\kern-.025em b}\kern-.08em
    T\kern-.1667em\lower.7ex\hbox{E}\kern-.125emX}}

\usepackage{cite}
\usepackage[latin1]{inputenc}
\usepackage{amsmath}
\usepackage{amsfonts}
\usepackage{amssymb}
\usepackage{amsthm}
\usepackage{color}
\usepackage{url}
\usepackage{ulem}
\usepackage{comment}
\usepackage[pdftex]{graphicx}
\usepackage{paralist}

\usepackage{subcaption}
\usepackage{dsfont}
\usepackage{hyperref}

\usepackage{csquotes}
\usepackage{balance}
\usepackage[nolist,nohyperlinks]{acronym}
\usepackage{array}
\usepackage[switch]{lineno}
\usepackage{algorithm}
\usepackage{algpseudocode}

\usepackage{times}
\usepackage{bm}

\newtheorem{theorem}{Theorem}[section]
\newtheorem{corollary}{Corollary}[theorem]
\newtheorem{proposition}{Proposition}[section]

\newtheorem{lemma}{Lemma}

\DeclareMathOperator*{\argmax}{arg\,max}
\DeclareMathOperator*{\argmin}{arg\,min}

\newcommand{\Cret}[0]{c_f}
\newcommand{\El}[0]{\Lambda}
\newcommand{\knn}[0]{k\mathrm{NN}}

\newcommand{\norm}[1]{\left\lVert#1\right\rVert}
\newcommand{\OMA}[0]{\mathrm{OMA}}

\newcommand{\Regret}[2]{{\mathrm{Regret}_{T, #1}} (#2)}
\newcommand{\convX}[0]{\mathrm{conv}(\mathcal{X})}

\renewcommand{\vec}[1]{\bm{\mathbf{#1}}}
\newcommand{\x}[0]{\vec{x}}
\newcommand{\y}[0]{\vec{y}}
\newcommand{\z}[0]{\vec{z}}

\newcommand{\DepRound}[0]{\textproc{DepRound}}
\newcommand{\Round}[0]{\textproc{CoupledRounding}}
\newcommand{\Simplify}[0]{\textproc{Simplify}}
\newcommand{\lru}[0]{\textproc{LRU}}
\newcommand{\lfu}[0]{\textproc{LFU}}

\newcommand{\simlru}[0]{\textproc{SIM-LRU}}
\newcommand{\rndlru}[0]{\textproc{RND-LRU}}
\newcommand{\clslru}[0]{\textproc{CLS-LRU}}
\newcommand{\qcache}[0]{\textproc{QCache}}

\newcommand{\knntext}{$k$NN}
\newcommand{\acai}{\textproc{A\c{C}AI}}
\newcommand{\roundInterval}{M}
\newcommand{\BigO}[1]{\ensuremath{\operatorname{\mathcal O}\left(#1\right)}}

\newcommand{\card}[1]{\left|#1\right|}
\usepackage{mathtools}

\DeclarePairedDelimiter\floor{\lfloor}{\rfloor}

\usepackage{enumitem}
\usepackage{soul}
\usepackage{bbm}
\newcommand{\removeextra}{1}
\newcommand{\extended}{1}
\newcommand{\Old}[1]{}
\newcommand{\New}[1]{#1}

\newcommand{\nnew}[1]{\textcolor{black}{#1}}
\newcommand{\oold}[1]{\textcolor{blue}{}}

\usepackage{array}
\newcolumntype{C}[1]{>{\raggedright\arraybackslash}p{#1}}

\newcommand{\new}[1]{\ifnum\removeextra=0#1\else#1\fi}
\newcommand{\old}[1]{}

\begin{document}
\title{Ascent Similarity Caching \\ with Approximate Indexes}

\author{\IEEEauthorblockN{Tareq Si Salem}
\IEEEauthorblockA{{Universit\'e C\^ote d'Azur, Inria}\\
\texttt{ tareq.si-salem@inria.fr}}
\and
\IEEEauthorblockN{Giovanni Neglia}
\IEEEauthorblockA{{Inria, Universit\'e C\^ote d'Azur}\\
\texttt{giovanni.neglia@inria.fr}}
\and

\IEEEauthorblockN{Damiano Carra}
\IEEEauthorblockA{{University of Verona}\\
\texttt{damiano.carra@univr.it}}
}

\maketitle
\thispagestyle{plain}
\pagestyle{plain}

\begin{abstract}
    Similarity search is a key operation in multimedia retrieval systems and recommender systems, and it will play an important role also for future machine learning and augmented reality applications.
    When these systems need to serve large objects with tight delay constraints, edge servers close to the end-user can operate as \emph{similarity caches} to speed up the retrieval.
    In this paper we present \acai, a new similarity caching policy {which} improves on the state of the art  by  using (i) an (approximate) index for the whole catalog to decide which objects to serve locally and which to retrieve from the remote server, and (ii) a mirror ascent algorithm to update the set of local objects with strong guarantees even when the request process does not exhibit any statistical regularity. 
\end{abstract}
\section{Introduction}
\label{sec:introduction}

Mobile devices can enable \oold{a} rich interaction with the environment people are in. Applications such as object recognition or, in general, augmented reality, require {to process} and {retrieve}  in real time \oold{a set of}   information related to the content visualized by the camera. The logic behind such applications is very complex: although mobile devices' computational power and memory constantly increase, they may not be sufficient to run these sophisticated logics, especially considering the associated energy consumption. On the other hand, sending the data to the cloud to be processed introduces additional delays that may be undesirable or simply intolerable~\cite{chen15glimpse}.
Edge Computing \cite{zhou2019edge,venugopal2018shadow} solves this dichotomy by providing distributed computational and memory resources close to the users. Mobile devices may pre-process locally the data and send the requests to the closest edge server, which runs the application logic and provides quickly the answers.

Augmented reality applications often require {to identify}  similar objects: for example\nnew{,} an image (or an opportune encoding of it) can be sent as a query, and the application logic finds similar objects to be returned to the user~\cite{drolia2017cachier,drolia2017precog,guo2018foggycache,venugopal2018shadow}. For instance, a recommendation system may suggest similar  products to a user browsing shop windows in a mall, or similar artists to a user enjoying \oold{a} street \oold{artwork} \nnew{art}. The search for similar objects is based on  a $k$-nearest neighbor (\knntext) search in an opportune metric space~\cite{bellet2015metric}.
The flexibility of \knntext{} search comes at the cost of (i) high computational complexity in case of high dimensional spaces, and (ii) large memory required to store the instances. The first issue has been solved in recent years with a set of techniques used to index the collection of objects that provide \emph{approximate answers} to \knntext{} searches, i.e., they trade accuracy for speed. Searches over large catalogs (billions of entries) in high dimensional spaces may be executed now in less than a millisecond~\cite{johnson2019billion}. Still, the issue \oold{about} \nnew{of} the memory required to store the objects \oold{remain} \nnew{persists}, especially in a distributed edge computing scenario, as edge servers have limited memory resources compared to the cloud.

\oold{The selection of which objects a specific edge server should maintain remains an open issue.}
\nnew{Due to memory constraints, edge servers can be forced to store a subset of the objects in the catalog, but object selection is not an easy task.}
 Requests coming from the users often exhibit spatial and temporal correlation---e.g., the same augmented reality application  will recover different information in different areas, and this information can change over time as the environment changes and users' interests evolve. This observation suggests that we may use the request pattern to drive the object selection. In other words, the edge server can be viewed as a \emph{cache} that contains the set of objects
needed to \oold{efficiently} \nnew{adequately} respond to local requests while avoiding forwarding them to the cloud.

In this paper, we study how to optimize the memory usage of the edge server for similarity searches. To this aim, we consider the \emph{costs} associated \oold{to} \nnew{with} the replies, which capture both the quality of the reply (that is how similar/dissimilar to the request the objects provided are), as well as system costs like the delay experienced, the load on the server or \oold{on} the network.
\oold{The aim of our study is} \nnew{Our study aims} to design an online algorithm to \emph{minimize such costs}. We provide the following contributions:
\begin{itemize}
  \item We formulate the problem of \knntext{} optimal caching taking into account both dissimilarity costs and system costs.
  \item We propose a new similarity caching policy, \acai{} (Ascent Similarity Caching with Approximate Indexes), that \Old{1)} \New{(i)} relies on fast, approximate similarity search indexes to decide which objects to serve from the local datastore and which ones from the remote repository and \Old{2)} \New{(ii)} uses an online mirror ascent algorithm to update the \oold{cache} \nnew{cached} content \oold{in order} to minimize the total service cost. \acai{} offers strong theoretical guarantees without any assumption on the traffic arrival pattern.
  \item We compare our solution with state-of-the-art algorithms for similarity caching and show that \acai{} consistently  improves over all of them under realistic traces.
\end{itemize}

\noindent
The remainder of the paper is organized as follows: we present similarity caches in Sec.~\ref{sec:bg_similarity} and other relevant background in Sec.~\ref{sec:bg_other}. We introduce \acai{} in Sec.~\ref{sec:acai} and our experimental results in Sec.~\ref{sec:experiments}. \ifnum\extended=0 The reviewer may consult an extended version available online~\cite{sisalem20techrep}. It contains additional results and more detailed proofs.\fi The notation used across the paper is summarized in Table~\ref{tab:notations_definitions}. The proofs and other technical details can be found in the Supplementary material. \New{This work is an extension of our previous work~\cite{salem2021accai}. In particular, we provide \nnew{{(i)}} the regret bound when the state of the system is only updated every $M$ requests, {(ii)} a new rounding mechanism, and {(iii)} \nnew{the} \nnew{complete} proofs.}

 \section{Similarity Caches}
\label{sec:bg_similarity}
Consider a remote server that stores a catalog of objects $\mathcal{N} \triangleq \{1,2, \dots, N\}$. A similarity search request $r$ aims at finding the $k$ objects $o_1, o_2, \dots, o_k \in \mathcal{N}$ that are most similar to $r$  given an application-specific definition of \emph{similarity}. To this purpose, similarity search systems rely on a function $c_d(r,o) \in \mathbb R_{\geq 0}$, \oold{that} \nnew{which} quantifies the dissimilarity of a request $r$ and an object $o$. We call such \nnew{a} function the \emph{dissimilarity cost}.

In practice, objects and requests are mapped to vectors in $\mathbb R^d$ (called \emph{embeddings}), so that the dissimilarity cost can be represented as (a function of) a selected distance between the corresponding embeddings.
For instance, in \nnew{the} case of images, the embeddings could be a set of
descriptors like SIFT \cite{lowe1999object}, \oold{or} ORB \cite{rublee2011orb}, or the set of activation values at an intermediate layer of a neural network~\cite{hinton2006reducing, lin2016learning}. 
Examples of commonly employed distances are the  $p$-norm, Mahalanobis, \nnew{and} \oold{or} cosine \oold{ similarity}  distances.






The server replies to each request $r$ with the $k$ most similar objects in the catalog $\mathcal N$. As the dissimilarity is captured by the distance in the specific metric space, these objects are also the $k$ closest objects (neighbors) in the catalog to the request~$r$ ($\knn(r,\mathcal N)$).\footnote{
More precisely, these are the $k$ objects whose embeddings are closer to the embedding of $r$. From now on we identify objects \oold{and} \nnew{with} their embeddings.
} The mapping translates the similarity search problem \oold{in} \nnew{into}  a \knntext{} problem~\cite{yianilos1993data,navarro2002searching}.

We can also associate a dissimilarity cost to the reply provided by \nnew{a} server (e.g., by summing the dissimilarity costs for all objects in $\knn(r,\mathcal N)$). This cost depends on the catalog $\mathcal N$ and we do not have control \oold{on} \nnew{over} it. In addition, there is a \emph{fetching cost} to retrieve those objects. The fetching cost captures, for instance, the extra load experienced by the server or the network to provide the objects to the user, the delay experienced by the user\nnew{,} or a mixture of those costs. 

In the Edge Computing scenario 
we consider, we can reduce the fetching cost by storing at the edge server a subset of the catalog $\mathcal N$, i.e., the edge server works as a \emph{cache}. When answering to a request, the cache could provide just some of the \nnew{$k$~closest} objects \oold{the server would provide}  \nnew{(those stored locally) and retrieve the others from the server}. The seminal papers~\cite{falchi2008metric, pandey2009nearest} proposed a different use of the cache:  the cache may reply to a request using a local subset of objects that are potentially \emph{farther} than the \oold{true} closest neighbors to reduce the fetching cost while increasing---hopefully only slightly---the dissimilarity cost. They named such cache a \emph{similarity cache}. The envisaged applications were content-based image retrieval~\cite{falchi2008metric} and contextual advertising~\cite{pandey2009nearest}. But, as recognized in \cite{garetto2020similarity}, the idea has been rediscovered \oold{a number of} \nnew{several} times under different names for different applications: semantic caches for object recognition~\cite{drolia2017cachier,drolia2017precog,guo2018foggycache,venugopal2018shadow}, soft caches for recommender systems \cite{Sermpezis18,costantini20}, approximate caches for fast machine learning inference \cite{crankshaw2017clipper}. 

A common assumption in the existing literature is that the cache can only store $h$ objects and the index needed to manage them has essentially negligible size. \new{We also maintain this assumption, which is justified in practice when objects have sizes of a few tens of kilobytes~(see the quantitative examples in Sec.~\ref{sec:bg_other})}. 

\noindent
{\bf Caching policies.}
The performance of the cache depends heavily on which objects the cache stores. Among the papers mentioned above, many (e.g., \cite{Sermpezis18, costantini20}) consider the offline object placement problem: a set of objects is selected \oold{on the basis of} \nnew{based on}  historical information about object popularity and prefetched in the cache. But object popularity can be difficult to estimate and can change over time, \oold{specially} \nnew{especially} at the level of small geographical areas (as in the case of areas served by an edge server)~\cite{paschos16}.
Other papers~\cite{crankshaw2017clipper, drolia2017cachier, drolia2017precog, guo2018foggycache, guo2018potluck, venugopal2018shadow} present more a high-level view of the different components of the application system, without specific contributions in terms of cache management policies (e.g., they apply minor changes to exact caching policies like \lru{} or \lfu). 
Some recent papers~\cite{garetto2020similarity, zhou20,sabnis21} propose online caching policies that try to minimize the total cost of the system (the sum of the dissimilarity cost and the fetching cost), also in a networked context~\cite{zhou20, GARETTO2021108570}, but their schemes apply only to the case $k=1$, which is of limited practical interest.

To the best of our knowledge, the only dynamic caching policies conceived to manage the retrieval of $k>1$ similar objects are  \simlru, \clslru, and \rndlru{} proposed in~\cite{pandey2009nearest} and \qcache{} proposed in \cite{FALCHI2012803}.
Next, we describe in detail these policies to highlight \acai's differences and novelty.

All these policies maintain an ordered list of key-value pairs where the key is a previous request and the value is the set of $k'$ closest objects to the request in the catalog (in general $k'\ge k$). The cache, whose size is $h$, maintains a set of $h/k'$ past requests. This approach allows to decompose the potentially expensive search for close objects in the cache (see Sec.~\ref{sec:bg_other}) in two separate less expensive searches on smaller sets. Upon \New{the} arrival of a request $r$, the cache identifies the $l$ closest requests to $r$ among the $h/k'$ in the cache. Then, it merges their corresponding values and looks for the $k$ closest objects to $r$ in this set including at most $l \times k'$ objects. 
\nnew{As the cache has no knowledge about the catalog at the server, it cannot compare the quality (i.e., the dissimilarity cost) of the local answer with the quality of the answer the server can provide. It relies then on heuristics (detailed below) to decide if the local answer is good enough. 
If this is the case, }
then an \emph{approximate hit} occurs and the answer is provided to the user, otherwise the request $r$ is forwarded to the server that needs to provide all $k$ closest objects.

The cache state is updated following \oold{a} \nnew{an} \lru-like approach: upon an approximate hit, all key-value pairs that contributed to the answer are moved to the front of the list; upon a miss, the new key-value pair provided by the server is stored at the front of the list, and the pair at the end of the list is evicted.

This operation is common to  \simlru, \clslru, \rndlru, and  \qcache. They differ in the choice of the parameters $k'$ and $l$ and in the way to decide between an approximate hit and a miss. 
\New{We emphasize that the parameters $k'$ and $l$ are only required by the LRU-like policies and do not play any role in \acai's workflow.}

\simlru{} considers $k'\ge k$ and $l=1$. Upon a request for $r$, \simlru{} selects the closest request in the cache and decides for an approximate hit (resp.~a miss) if their dissimilarity is smaller (resp. larger) than a given threshold $C_\theta \in \mathbb R_{\geq 0} $. 
Every stored key $r'$ \emph{covers} then a hypersphere in the request space with radius $C_\theta$.
\simlru{} has the  property that no two keys in the cache have a dissimilarity cost lower than $C_\theta$, but the corresponding hyperspheres may still intersect. 

\clslru{} \cite{pandey2009nearest} is a variant of \simlru, that can update the stored keys (the centers of the hyperspheres) and push away intersecting hyperspheres to cover the largest possible area of the request space. To this purpose, \clslru{} maintains the history of requests served at each hypersphere and, upon an approximate hit, moves the center to the object that minimizes the distance to every object within the hypersphere's history. When two hyperspheres overlap, this mechanism drives their centers apart, which in turn reduces the overlapping region. 

\rndlru~\cite{pandey2009nearest} is a random variant of \simlru{} that determines the request $r$ to be a miss with a probability that is increasing with the dissimilarity \oold{cost} between $r$ and the closest request in the cache.

Finally, \qcache{}~\cite{FALCHI2012803} considers $k'=k$ and $l>1$. The policy decides if the $k$ objects selected from the cache are an approximate hit if \Old{(1)}\New{(i)}~at least two of them would have been provided also by the server---a sufficient condition can be obtained from geometric considerations---or \Old{(2)} \New{(ii)}~the distribution of distances of the $k$ objects from the request looks similar to the distribution of objects around the corresponding request for other stored key-value pairs.

These policies share potential inefficiencies: (i) the sets of closest objects to previous queries are not necessarily disjoint (but \clslru{} tries to reduce their overlap) and then the cache may store less than $h$ distinct objects; (ii) the two-level search may miss some objects in the cache that are close to~$r$, but are indexed by requests that are not among the $l$ closest requests to $r$; (iii) the policy takes into account the dissimilarity costs at the caches but not at the server; (iv) objects are served in \Old{block}\New{bulk}, all from the cache or all from the server, without the flexibility of a per-object choice.
As we are going to see, \acai{} design prevents such inefficiencies by exploiting new advances in efficient approximate \knntext{} search algorithms, which allows us to abandon the key-value pair indexing and to estimate the dissimilarity costs at the server. Also \acai{} departs from the \lru-like cache updates, considering gradient update schemes inspired by online learning algorithms \cite{paschos2019learning}.

\section{Other Relevant Background}
\label{sec:bg_other}

\noindent
{\bf Indexes for approximate \knntext{} search.}
Indexes are used to efficiently search objects in a large catalog. In \nnew{the} case of \knntext{}, one of the approaches is to use tree-based data structures. Unfortunately, in high dimensional spaces, e.g., $\mathbb R^d$ with $d > 10$, the computational cost of such \nnew{a} search is comparable to a full scan of the collection \cite{weber1998quantitative}. \emph{Approximate Nearest Neighbor}  search techniques trade accuracy for speed and provide $k$ points close to the query, but not necessarily the closest, sometimes with a guaranteed bounded error. Prominent examples are the solutions based on locality\nnew{-}sensitive hashing  \cite{andoni2006near}, product quantization \cite{johnson2019billion,babenko2014inverted}, and graphs~\cite{malkov2018efficient}. Despite being approximate, these indexes are in practice very accurate, as \oold{showed} \nnew{shown} over different benchmarks in \cite{aumuller2017ann}. 

As we are going to describe, \acai{} employs two approximate indexes (both stored at the edge server): one for the content stored in the cache, and one for the whole catalog~$\mathcal N$ stored in the remote server. For the former, since cache content varies over time, we rely on a graph-based solution, such as HNSW \cite{malkov2018efficient}, that supports dynamic (re-)indexing with no speed loss.
On various benchmarks~\cite{aumuller2017ann}, HNSW \oold{results} \nnew{is} the fastest index, and it is able to answer a $100$NN query over a dataset with 1 million objects in a 128-dimensional space in less than 0.5 ms with a recall greater than 97\%.\Old{{Experiments on a 4-core Intel Core i7-4790, 32 GB RAM, 3.6 GHz.}} 
As for the memory footprint, a typical configuration of the HNSW index requires $O(d)$ bytes per object\oold{s}, where $d$ is the number of dimensions. For instance, in \nnew{the} case of $d=128$ dimensional vectors, the memory required to index 10 million objects is approximately 5 GB. 
As the server catalog changes less frequently (e.g., \New{contextual advertising applications~\cite{pandey2009nearest}, and image retrieval applications~\cite{FALCHI2012803}}), \acai{} can index it using approaches with a more compact object representation like FAISS \cite{johnson2019billion}. FAISS is slightly {slower} than HNSW and does not support fast re-indexing if the catalog changes, but it can manage a much larger set of objects. With a dataset of 1 billion objects, FAISS provides an answer in less than 0.7 ms per query, using a GPU~\cite{johnson2019billion}.\Old{{Experiments on 2x2.8GHz Intel Xeon E5-2680v2, 4 Maxwell Titan X GPUs, CUDA 8.0.}} \New{Practically, the global catalog index can be fully reconstructed whenever a given percentage of the catalog changes.}
\nnew{This operation can be done in parallel to the normal cache operation. Once the catalog index is modified, the cache can remove the objects that do not appear anymore in the catalog and allocate the corresponding space to other objects.}
\oold{When a new global catalog index is reconstructed, the state of \acai{} is modified accordingly by redistributing the total mass on the contents that disappeared on the new introduced contents.}
As for the memory footprint, for a typical configuration (IVFPQ), FAISS is able to represent an object with 30 bytes (independently of $d$): only 3 GB for a dataset with 100 million objects!

\new{ Summing up our numerical example, if each object has size 20~KB, an edge server with \acai{} storing locally 10~million objects from a catalog with 100~million objects, needs 200~GB for the objects and only 8~GB
for the two indexes. The larger the objects, the smaller the index\New{es'}  footprint: for example, when the server has a few Terabytes of disk space to store large multimedia objects, the indexes' size can be ignored.}

\vspace{1mm}
\noindent
{\bf Gradient descent approaches.}
Online caching policies based on gradient methods have been studied in the stochastic request setting for exact caching, with provable performance guarantees, \cite{ioannidis2010distributed,ioannidis2016adaptive}. More recently, the authors of~\cite{sabnis21} have proposed a gradient method to refine the allocation of objects stored by traditional similarity caching policies like \simlru. Similarly, the reference \cite{zhou20} considers a heuristic based on the gradient descent/ascent algorithm to allocate objects in a network of similarity caches. In both papers, the system provides a single similar \oold{content} \nnew{object} ($k=1$). 
{
A closely related recent work~\cite{sisalemmedcomnet21} considers the problem of allocating different inference models that can satisfy users' queries at different quality levels. The authors propose a policy based on mirror descent, and provide guarantees under a general request process,  but their policy does not scale to a large catalog size.  
}

We deviate from these works by considering $k>1$, large catalog size,  and the more general family of online mirror ascent algorithms (of which the usual gradient ascent method is a particular instance).  Also our policy provides strong performance guarantees under a general request process, where requests can even be selected by an adversary. Our analysis relies on results from online convex optimization~\cite{shalev2011online} and is similar in spirit \old{to what done in~\cite{paschos2019learning} for exact caching using the classic gradient method.}\new{to what \nnew{was} done for exact caching using the classic gradient method in~\cite{paschos2019learning} and mirror descent in~\cite{sisalem21icc}. }\New{Two recent papers~\cite{sisalem21arxiv,9517925} pursued this line of work taking into account update costs for a single exact cache.
}


 \section{\acai{} Design}
\label{sec:acai}
\New{\acai{} design is summarized in Fig.~\ref{fig:acaidesign}.}
\begin{figure*}[t]
    \centering
    \includegraphics[width =.9 \linewidth]{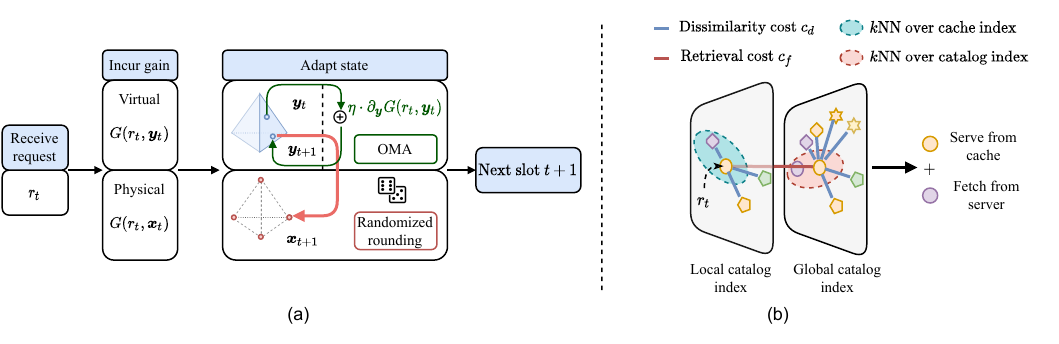}
    \caption{\New{Subfigure~(a) illustrates \acai's state adaptation. A time slot is initiated when a request $r_t$ is received. A  virtual (fictitious) gain $G(r_t, \y_t)$ and a physical gain $G(r_t, \x_t)$ are incurred. The virtual cache adapts its fractional state by calling Online Mirror Ascent to obtain a new state $\y_{t+1} \in \convX$ employing the subgradient of the virtual gain $\partial_{\y} G(r_t, \y_t)$, and the new state is randomly rounded to a valid cache state $\x_{t+1} \in \mathcal{X}$.
    Subfigure~(b) depicts how \acai{} employs the two indexes (local catalog index and global catalog index). \oold{An} Approximate \knntext{} queries are performed on each index, and the contents with the least overall costs are selected. } }
    \label{fig:acaidesign}
\end{figure*}

\subsection{Cost Assumptions}
Many of the similarity caching policies proposed in the literature (including \simlru, \clslru, \rndlru, and \qcache) have not been designed with a clear quantitative objective, but with the qualitative goal of significantly reducing the fetching cost without increasing too much the dissimilarity cost. Because of such vagueness, the corresponding papers do not make clear assumptions about the dissimilarity costs and the fetching costs.
On the contrary, \acai{} has been designed to minimize the total cost of the similarity search system and we make explicit the corresponding hypotheses.

\begin{table}[htbp]

  \caption{{Notation Summary } \label{tab:notations_definitions}  }
  \vspace{-1em}
  \begin{center}
    \begin{footnotesize}
      \begin{tabular}{|C{.055\textwidth}C{.38\textwidth}|}
        \hline
                                     & \textbf{Notational Conventions }                                                                                                        \\
        \hline

        $\mathbbm{1}_{\chi}$         & Indicator function set to 1 when condition $\chi$ is true                                                                               \\
        $[n]$                        & Set of integers $\{1,2, \dots, n\}$                                                                                                     \\
        $\mathrm{conv}(S)$           & Convex hull of a set $S$                                                                                                                \\
        \hline
                                     & \textbf{System Model}                                                                                                                   \\
        \hline
        $\mathcal{N}$                & Catalog of $N$ objects                                                                                                                  \\
        $\mathcal{U}$                & Augmented catalog of $2N$ objects                                                                                                       \\
        $x_i$                        & 0-1 indicator variable set to 1 when $i \in \mathcal{N}$ is cached, and $x_i = 1-x_{i-N}$ for $i \in \mathcal{U} \setminus \mathcal{N}$ \\
        $h$                          & Cache capacity                                                                                                                          \\
        $r \, / \,\mathcal{R}$       & Request / Request set                                                                                                                   \\
        $c_f$                        & Retrieval cost                                                                                                                          \\
        $c_d(r,o)$                   & Dissimilarity cost of serving object $o$ to request $r$                                                                                 \\
        $c(r,o)$                     & Overall cost of serving object $o$ to request $r$                                                                                       \\
        $\pi^r$                      & Permutation of the elements of $\mathcal{U}$, where $\pi^r_i$ gives the $i$-th closest object to $r$                                    \\
        $\alpha^r_i$                 & Cost difference between the $(i+1)$-th smallest cost and the $i$-th smallest cost when serving request $r$                              \\
        $K^r$                        & The order of the largest possible cost when $r$ is requested.                                                                           \\
        $\knn(r, S)$                 & Set of $k$ closest objects to $r$ in $S \subset\mathcal{N}$ according to $c(r,\, \cdot\,)$                                              \\
        $\x \, / \, \mathcal{X}$     & Cache state vector / Set of valid cache states                                                                                          \\
        $t \, / \,T$                 & Time slot / Time horizon                                                                                                                \\
        $C(r,\x)$                    & Total cost to serve request $r$ under cache allocation $\x$                                                                             \\
        $G(r, \x) $                  & Total caching gain to serve request $r$ under cache allocation $\x$                                                                     \\
        $\mathcal{C}_{\text{UC}, T}$ & Update cost of the system over time horizon $T$                                                                                         \\
        $G_T(\x)$                    & Time-averaged caching gain                                                                                                              \\
        \hline
                                     & \textbf{\acai{}}                                                                                                                        \\
        \hline
        $\Phi$                       & Mirror map                                                                                                                              \\
        $\mathcal{D}$                & Domain of the mirror map                                                                                                                \\
        $\y$                         & Fractional cache state                                                                                                                  \\
        $\eta$                       & Learning rate                                                                                                                           \\
        $\vec g_t$                   & Subgradient of $G(r_t,\y)$ at point $\y_t$                                                                                              \\
        $\prod_{S}^\Phi(\,\cdot\,)$  & Negative entropy Bregman projection onto the set $S$                                                                                    \\
        $M$                          & Freezing period                                                                                                                         \\
        $c_d^k$                      & Upper bound on the dissimilarity cost of the k-th closest object  for any request                                                       \\
        $\psi$                       & Static optimum discount factor                                                                                                          \\
        \hline
      \end{tabular}
    \end{footnotesize}
  \end{center}
\end{table}

Our main assumption is that all costs are additive.\footnote{
\nnew{In Sec.~\ref{sec:cache_state_service_gain}, we discuss to which extent this assumption can be removed.}
} The function $c_d(r,o)$ introduced in Sec.~\ref{sec:bg_similarity} quantifies the dissimilarity of the object $o$ and the request $r$. Let $\mathcal A$ be the set of objects in the answer to request $r$. It is natural to consider as dissimilarity cost of the answer $\sum_{o \in \mathcal A} c_d(r,o)$.

In addition, if fetching a single object from the server incurs a cost $c_f \in \mathbb R_{> 0}$, the fetching cost to retrieve $m$ objects is $m \times c_f$. This is an obvious choice when $c_f$ captures server or network cost. When $c_f$ captures the delay experienced by the user, then summing the costs is equivalent to consider the round trip time negligible in comparison to the transmission time, which is justified for large multimedia objects. It is easy to modify \acai{} to consider the alternative case when the fetching cost does not depend on how many objects are retrieved.
Finally, as common in other works~\cite{garetto2020similarity,sabnis21}, we assume that both the dissimilarity cost and the fetching cost can be directly compared (e.g., they can both be converted \oold{in} \nnew{into} dollars).
Under these assumptions, when, for example, the $k$ nearest neighbors in $\mathcal N$ to the query $r$ ($\knn(r,\mathcal N)$) are retrieved from the remote server, the total cost experienced by the system is
$\sum_{o \in \knn(r,\mathcal N)} c_d(r,o) + k c_f$.
\subsection{Cache Indexes}
\acai{} departs from the key-value indexes of most \nnew{of} the similarity caching policies. As discussed in Sec.~\ref{sec:bg_similarity}, such \New{an} approach was essentially motivated by the need to simplify \knntext{} searches by performing two searches on smaller datasets (the set of keys first, and then the union of the values for $l$~keys), and may lead to potential inefficiencies including sub-utilization of the available caching space.

The two-level search implemented by existing similarity caching policies can be seen as a na\"ive way to implement an approximate \knntext{} search on the set of objects stored locally (the \emph{local catalog} $\mathcal C$).
Thanks to the recent advances in approximate \knntext{} searches (Sec.~\ref{sec:bg_other}), we have now better approaches to search through large catalogs with limited memory and computation requirements. We assume then that the cache maintains two indexes supporting \knntext{} searches: one for the local catalog \new{(the objects stored locally)} and one for the remote catalog \new{(the  objects stored at the server)}. A discussion about which approximate index is more appropriate for each catalog is in Sec.~\ref{sec:bg_other}.

The local catalog index allows \acai{} to (i) fully exploit the available space (the cache stores at any time $h$ objects and can perform a \knntext{} search on all of them), (ii) potentially find closer objects in comparison to the non-optimized key-value search. Instead, the remote catalog index allows \acai{} to evaluate what objects the server would provide as \nnew{an} answer to the request, and then to correctly evaluate which objects should be served locally and which one should be served from the server/ as we are going to describe next.

\subsection{Request Serving}
\label{sec:request_serving}
Differently from existing policies, \acai{} has the possibility to compose the answer using both local objects and remote ones.
Upon a request $r$, \acai{} uses the two indexes to find the closest objects from the local catalog $\mathcal C$ and \oold{from} the remote catalog $\mathcal N$. We denote the set of objects identified by these indexes as $\knn(r,\mathcal C)$ and $\knn(r,\mathcal N)$\nnew{,} respectively. \acai{} composes the answer  $\mathcal A$ by combining the objects with the smallest costs in the two sets. For an object $o$ stored locally ($o \in \mathcal C$), the system only pays $c_d(r,o)$; for an object $o$ fetched from the remote server ($o \in \mathcal N\setminus \mathcal C$), the system pays $c_d(r,o)+c_f$. The total cost experienced is
\begin{align}
  \label{e:answer_cost}
  C(r,\mathcal A) & \triangleq \hspace{-1.5em}\sum_{ o \in \mathcal A \cap \knn(r, \mathcal C)} c_d(r,o)  + \sum_{ o \in \mathcal A \setminus \knn(r, \mathcal C)} \left( c_d(r,o) + c_f\right).
\end{align}

The answer $\mathcal A$ is determined by selecting $k$ objects that minimize the total cost, that is
\begin{align}
  \label{e:answer}
  \mathcal A = \argmin_{\substack{{\mathcal B \subset \left(\knn(r,\mathcal C) \cup \knn(r,\mathcal N)\right)} \\
  {|\mathcal B| = k }}} C(r,\mathcal B).
\end{align}

\subsection{Cache State and Service Cost/Gain}
\label{sec:cache_state_service_gain}
In order to succinctly present how \acai{} updates the local catalog and its theoretical guarantee, it is convenient to express the cost in~\eqref{e:answer_cost} as \new{a} function of the current cache state and replace the set notation with a vectorial one.

First, we define the \emph{augmented catalog} $\mathcal{U} \triangleq \mathcal{N} \cup \{N+1, N+2, \dots, 2N\}$ and define the new costs
\begin{align}
\label{e:3}
  c(r,i)= \begin{cases}
    c_d(r,i),             & \text{if } i \in \mathcal{N},                       \\
    c_d(r,i - N) + \Cret, & \text{if } i \in \mathcal{U} \setminus \mathcal{N}.
  \end{cases}
\end{align}
Essentially, $i$ and $i+N$ (for $i \in \{1, \dots, N\}$) correspond to the same object, with $i$ capturing the cost when the object is stored at the cache and $i+N$ capturing the cost when it is stored at the server.
From now on, when we talk about the closest objects to a request, we are considering $c(\cdot, \cdot)$ as the distance.

\New{Note that \acai{} can easily be modified to account for heterogeneous retrieval costs by modifying Eq.~\eqref{e:3} and replacing $c_f$ by an object\nnew{-}dependent retrieval cost $c_{f,i}$ for every object $ i \in \mathcal{N}$.  Moreover, we assume that the fetching cost and dissimilarity cost are added together linearly in the objective, however our model can capture the scenario where the cost is not necessarily additive in $c_d$ and $c_f$ by a redefinition of the second line in Eq.~\eqref{e:3}. In particular, the algorithm only requires the existence of \nnew{an arbitrary function $c(r,i)$} \oold{and the particular form of this function can be arbitrary.} \nnew{and the theoretical guarantees in Sec.~\ref{sub:guarantees} also hold under such modifications. We consider a simplified model to streamline the presentation.}
\oold{To streamline the presentation we considered this simplified model and the theoretical guarantees will also hold under such modifications. }}

It is also convenient to represent the state of the cache (the set of objects stored locally) as a vector $\x \in \{0,1\}^{2N}$, where, for $i \in \mathcal N$, $x_i=1$ (resp., $x_i=0$), if $i$ is stored  (resp., is not stored) in the cache, and we set $x_{i+N}=1-x_i$.\footnote{
  The vector $\x$ has \oold{clearly} redundant components, but such redundancy leads to more compact expressions in what follows.
} The set of valid cache configurations is given by:
\begin{align}
  \mathcal{X} \triangleq \left\{ \x \in \{0,1\}^{2 N}: \sum_{i \in \mathcal{N}} x_i = h,\, x_{j+N}= 1-x_j, \forall j \in \mathcal N \right\}.
\end{align}

For every request $r \in \mathcal{R}$ we define the sequence $\pi^r$ as the permutation of the elements of $\mathcal{U}$, where $\pi^r_i$ \oold{gives} is the $i$-th closest object to $r$ in $\mathcal{U}$ according to the costs $c(r,o), \forall o \in \mathcal{U}$.
The answer $\mathcal A$ provided by \acai{} (Eq.~\eqref{e:answer}) coincides with the first $k$ elements of $\pi^r$ for which the corresponding index in $\x$ is equal to $1$. The total cost to serve $r$ can then be expressed directly as a function of the cache state $\x$:
\begin{align}
  C(r, \x) & = \sum^{2N}_{i=1} c(r,\pi^r_{i}) x_{\pi^r_i} \mathbbm{1}_{\left\{\sum^{i-1}_{j=1} x_{\pi^r_j} < k\right\}} \label{eq:Ca_cost_function}, \forall \x \in \mathcal{X},
\end{align}
where $\mathbbm{1}_{\{\chi\}} = 1$ when the condition $\chi$ is true, and $\mathbbm{1}_{\{\chi\}} = 0$ otherwise.

Instead of working with the cost $C(r,\x)$, we can equivalently consider the \emph{caching gain} defined as the cost reduction due to the presence of the cache (as in \cite{shanmugam2013femtocaching,ioannidis2016adaptive,neglia2019swiss}):
\begin{equation}
  \label{e:gain}
  G(r,\x) \triangleq C(r, (\underbrace{0,0, \dots, 0}_{N}, \underbrace{1, 1, \dots, 1}_{N})) - C(r,\x),
\end{equation}
where the first term corresponds to the cost when the cache is empty (and then requests are entirely satisfied by the server).
The theoretical guarantees of \acai{} are simpler to express in terms of the caching gain (Sec.~\ref{sub:guarantees}). \New{Observe that the caching gain is zero for any cache state when the retrieval cost is null ($c_f = 0$), \nnew{e.g., the cache and the server are co-located. In this case, the cache would not provide any advantage.}
\oold{and this corresponds to the setting where the remote server and the cache are co-located and this scenario overrides the need of a cache.}  }

The caching gain has the following compact expression \ifnum\extended=1 (Supplementary material,~Sec.~I, Lemma~3\else \cite[Lemma~5]{sisalem20techrep}\fi):
\begin{align}
  G( r,\x) =\sum^{K^r - 1
  }_{i=1} \alpha^r_i \min\left\{k - \sigma^r_i, \sum^i_{j=1} x_{\pi^r_j} - \sigma^r_i \right\},
  \label{eq:gain_compact}
\end{align}
where
\begin{align}
  \sigma^r_i \triangleq  \sum^i_{j=1} \mathbbm{1}_{\left\{\pi^r_j  \in \mathcal{U} \setminus \mathcal{N}\right\}}, &   & \forall (i,r) \in \mathcal{U} \times \mathcal{R},
  \label{eq:definition_of_sigma}
\end{align}
\old{denotes the number of objects that are among the top $i$ less costly objects to serve $r$, even when served from the server}
$K^r$ is the value of the minimum index $i \in \mathcal{U}$ such that $\sigma^r_i = k$, and $\alpha_i^r \triangleq c(r,\pi^r_{i+1}) - c(r,\pi^r_{i}) \ge 0$.
\old{
The function $\min\left\{k- \sigma^r_i , \sum^i_{j=1} y_{\pi^r_j}- \sigma^r_i\right\}$ is concave $\forall \y \in \mathbb{R}^{2N}$. We conclude that the gain function $\mathcal{G}_{r}$ is a concave function over the domain $\convX$, where $\convX$ is the convex hull of the set of valid cache configurations $\mathcal{X}$.
}

\new{Let $\convX$ denote the convex hull of the set of valid cache configurations $\mathcal{X}$. We observe that $G(r,\y)$ is a concave function of variable $\y \in \convX$. Indeed, from Eq.~\eqref{eq:gain_compact}, $G( r,\y)$ is a linear combination, with positive coefficients, of concave functions (the minimum of affine functions in $\y$).}

 \subsection{Cache Updates}
We denote by $r_t \in \mathcal{R}$ the $t$-th request. The cache is allowed to change its state $\x_t \in \mathcal{X}$ to $\x_{t+1} \in \mathcal{X}$ in a reactive manner, after receiving the request $r_t$ and incurring the gain $G({r_t} ,\x_t)$. \acai{} updates its state $\x_t$ \oold{with the goal of greedily maximizing} \nnew{ to greedily maximize}  the gain.
\begin{algorithm}[!t]  
\begin{footnotesize}
 \hspace*{1.4em} \textbf{Input:} $\eta \in \mathbb{R}_+$, \textsc{RoundingScheme}
	\begin{algorithmic}[1]
	\Procedure {OnlineMirrorAscent}{}
	\State $ \y_1 \gets \underset{\y \in \convX \cap \mathcal{D}}{\arg\min} \,\Phi(\y); \x_1 \gets \DepRound(\y_{1})$
	\For{$t \gets 1,2,\dots,T$}\Comment{{Incur a gain $G({r_t} ,\y_t)$, and compute a subgradient  $ \vec{g}_t $ of $ G$ at point $\y_t$ (Supplementary material\nnew{,}~Sec. V\nnew{,} \old{Eq.~(52)}\new{\ifnum \extended=1Eq.~(52)\else\cite[Eq.~(63)]{sisalem20techrep}\fi}}})
		\State
		$\hat{\y}_t \gets \nabla \Phi (\y_t)$\Comment{{Map primal point to dual point}}
		\State
		$\hat{\z}_{t+1}   \gets \hat{\y}_t + \eta \vec{g}_t$\Comment{{Take gradient step in the dual}}
		\State
		$\z_{t+1}   \gets \left(\nabla \Phi\right)^{-1}(\hat{\z}_{t+1})$\Comment{{Map dual point to a primal point}}
		\State
		$\y_{t+1}\gets \prod_{\convX \cap \mathcal{D}}^\Phi(\z_{t+1})$\Comment{{Proj. new point onto feasible region}}
		
	\Statex\Comment{ Select a rounding scheme}
		 \If {\textsc{RoundingScheme} = \DepRound}
		 \If{ $M \, | \, t $} \Comment{\new{Round the fractional state} every $M$ requests}
		 \State $\x_{t+1} \gets \DepRound(\y_{t+1})$ 
		 
		 \EndIf
		 \ElsIf{\textsc{RoundingScheme} = \Round}
		 \State $\x_{t+1} \gets \Round(\x_t, \y_t, \y_{t+1})$
		 \EndIf
		\EndFor
		\EndProcedure

	\end{algorithmic}
		\end{footnotesize}
	\caption{Online Mirror Ascent ($\OMA$)}
	\label{algo:online_mirror_ascent}
\end{algorithm}
The update of the state $\x_t$ is driven from a continuous fractional state $\y_t \in \convX$, where $y_{t,i}$ can be interpreted as the probability to store object $i$ in the cache. 
At each request~$r_t$, \acai{} increases the components of $\y_t$ corresponding to the objects that are used to answer to $r_t$, and decreases the other components. This could be achieved by a classic gradient method, e.g.,~$\y_{t+1} = \y_t + \eta \mathbf g_t $, where $\mathbf g_t$ is a subgradient of $G(r_t, \y_t)$ and $\eta \in \mathbb{R}_+$ is the learning rate (or stepsize), but in \acai{} we consider a more general online mirror ascent update $\OMA$~\cite[Ch.~4]{bubeck2015convexbook} that is described in Algorithm~\ref{algo:online_mirror_ascent}.\footnote{
    Properly speaking $\OMA$, only refers to the update of $\y_t$ and does not include the randomized rounding schemes in lines 8--14.
} $\OMA$ is parameterized by the function $\Phi(\,\cdot\,)$, that is called the \emph{mirror map} (see~Supplementary material\nnew{,} Sec.~VI). If the mirror map is the squared Euclidean norm, $\OMA$ coincides with the usual gradient ascent method, but other mirror maps can be selected.
In particular, our experiments in Sec.~\ref{sec:experiments} show that the negative entropy map $\Phi(\y)=\sum_{i \in \mathcal{N}} y_i \log y_i$ \New{with domain $\mathcal{D} = \mathbb{R}^N_{>0}$} achieves better performance.

\subsection{Rounding the Cache Auxiliary State}
\label{a:sampling}
At every time slot $t \in [T]$, \acai{} can use the randomized rounding scheme \DepRound \cite{byrka2014improved} to generate a cache allocation $\x_{t+1} \in \mathcal{X}$ from $\y_{t+1}  \in \convX$, while still satisfying the capacity constraint at any time slot~$t$. The cache can fetch from the server the objects that are in $\x_{t+1}$ but not in $\x_{t}$.

As cache \oold{movements/}updates  introduce extra costs \oold{to} \nnew{for} the network operator, \DepRound{} could potentially cause extra update costs that grow linearly in time. To mitigate incurring large update costs, we may avoid updating the cache state at every time slot $t \in [T]$ by 
\emph{freezing} the cache physical state for $\roundInterval \in [T]$ time steps. \New{ In particular, we assume the update cost of the system to be proportional to the number of fetched files, which can be upper bounded by the $l_1$ norm of the state update: $\sum_{i \in \mathcal{N}} \max\{0, x_{t+1,i}-x_{t,i}\} \leq \norm{\x_{t+1} - \x_{t}}_1$. Hence, if we denote by $\mathcal{C}_{\mathrm{UC},T}$ the total update cost of the system over the time horizon $T$, then  we have
\begin{align}
    \label{eq:update_cost}
    \mathcal{C}_{\mathrm{UC},T} = \mathcal{O} \left(\sum^{T-1}_{t=1}\norm{\x_{t+1} - \x_t}_1\right).
\end{align}
When the cache state is refreshed after a call to the rounding scheme \DepRound{}, the incurred update cost is in the order of $\BigO{2h}$, and $\frac{\mathcal{C}_{\mathrm{UC},T}}{T} = \BigO{ \frac{2 h}{M}}$. Moreover, when $M = \Theta \left(T^\beta\right)$ for $\beta \in (0, 1)$ it holds $ \frac{\mathcal{C}_{\mathrm{UC},T}}{T} = \BigO{T^{-\beta}}, $ and for any $\epsilon >0$ and $T$ large enough  
\begin{align} &\frac{\mathcal{C}_{\mathrm{UC},T}}{T}\leq \epsilon.
\end{align}
The average update cost of the system is then negligible for large $T$.} 
The parameter $\roundInterval$ reduces cache updates at the expense of reducing the cache reactivity \New{(see Theorem~\ref{theorem:main})}. 

In some applications, it is possible to  slightly violate the capacity constraint with small deviations, as long as this is satisfied on average~\cite{lorido2014review,elastic1,carra2020elastic}. For example, there could be a monetary value associated to the storage reserved by the cache, and a total budget available over a target time horizon $T$. In this setting, the cache may violate momentarily the capacity constraint, as far as the total payment does not exceed the budget.

\Round{} (Algorithm~\ref{algo:coupling_scheme}) is a rounding approach which works under this relaxed capacity constraint and does not require \nnew{freezing the cache state for $M$ time slots}. 
\oold{state freezing with the parameter $M$.}
At time slot $t\in[T]$, the cache decides which files to cache through $N$ coin tosses, where the file $i \in \mathcal{N}$ is cached with probability $y_{t,i}$, and the random state obtained is $\x_t$. By definition, the expected value of the integral state is $\mathbb E[\x_t] = \y_t$. The probability that the cache exceeds its \nnew{target }storage capacity by $\delta h$ is given by the Chernoff bound \cite{MitzenmacherProbabilityandComputing} as: 
\begin{align}
    \mathbb P \left(\norm{\x_t}_1 > (1+\delta) h\right)<e^{\frac{-\delta^{2} h }{2}}, \delta \in (0,1], 
\end{align}
where   the $l_1$ norm is restricted to the first $N$ components of the vector, i.e., $\norm{\x}_1 \triangleq \sum_{i \in \mathcal{N}} |x_i|$. 
In the regime of large cache sizes $h \gg 1$, we observe from the Chernoff bound that the cache stores less than $(1+\delta)$ of its \nnew{target} capacity $h$ with high probability. 

\Old{We assume the update cost of the system to be proportional to the number of fetched files, which can be upper bounded by the $l_1$ norm of the state update: $\sum_{i \in \mathcal{N}} \max\{0, x_{t+1,i}-x_{t,i}\} \leq \norm{\x_{t+1} - \x_{t}}_1$. Hence, if we denote by $\mathcal{C}_{\mathrm{UC},T}$ the total update cost of the system over the time horizon $T$, we have
\begin{align}
    \mathcal{C}_{\mathrm{UC},T} = \mathcal{O} \left(\sum^{T-1}_{t=1}\norm{\x_{t+1} - \x_t}_1\right).
\end{align}
}
Theorem~\ref{theorem:coupling_scheme} (proof in Supplementary material\nnew{,} Sec.~VIII-A) shows that the expected movement of \Round{} is equal to the movement of the fractional auxiliary states $\{\y_t\}^T_{t=1}$. 
\begin{theorem}
\label{theorem:coupling_scheme}
If the input to Algorithm~\ref{algo:coupling_scheme} is sampled from a random variable $\x_t \in \{0,1\}^N$ with $\mathbb E[\x_t] =\y_t$, then we obtain as output an integral cache configuration $\x_{t+1} \in \{0,1\}^N$ satisfying $\mathbb E[\x_{t+1}] =  \y_{t+1}$ and $\mathbb E  \left[\norm{\x_{t+1} -\x_{t}}_{1}\right] =  {\norm{\y_{t+1} - \y_{t}}_{1}}$.
\end{theorem}
Moreover, the movement of the fractional states  
is negligible for large $T$:
\begin{theorem}
Algorithm \ref{algo:online_mirror_ascent}, configured with the negative entropy mirror map and learning rate $\eta = \BigO{\frac{1}{\sqrt{T}}}$, selects fractional cache states satisfying
\begin{align}
    \sum^{T-1}_{t=1} \norm{\y_{t+1} - \y_{t}}_1= \mathcal{O}(\sqrt{T}).
\end{align}
\label{theorem:omd_movement_cost}
\end{theorem}
The proof is in Supplementary material\nnew{,} Sec.~VIII-B.

Combining the two theorems and \eqref{eq:update_cost}, \New{ we also conclude that  the  expected average update cost of the system $\mathbb E \left[ \frac{\mathcal{C}_{\mathrm{UC},T}}{T} \right]$ is  negligible for large $T$.}
\Old{
for any $\epsilon > 0$ and $T$ large enough:
\begin{align}
      \mathbb E \left [ \frac{\mathcal{C}_{\mathrm{UC},T}}{T}\right ] \leq \epsilon.
\end{align}
The expected average update cost of the system is then negligible for large $T$. }

\begin{algorithm}[t]
\begin{footnotesize}
 \hspace*{1.4em} \textbf{Input:} $\x_{t}, \y_{t} , \y_{t+1}$ \Comment{$\x_t$ satisfies $\mathbb E [\x_t] = \y_t$}
	\begin{algorithmic}[1]
	\Procedure {CoupledRounding}{}
	
	\State $\delta \gets \y_{t+1} - \y_{t}$ \Comment{Compute the change in distribution}
		\For{$ i \in \mathcal{N}$}
		\If {$\left(x_{t,i} = 1 \right)\land \left(\delta_i < 0\right)$} 
		 \State $ x_{t+1,i} \gets 0$ w.p. $- \frac{\delta_i}{y_{t,i}}$, and $x_{t+1,i} \gets 1$ w.p. $\frac{y_{t,i}+\delta_i}{y_{t,i}}$
		\ElsIf{$\left(x_{t,i} = 0\right) \land \left(\delta_i > 0\right)$}
		 \State $ x_{t+1,i} \gets 0$ w.p. $\frac{1 - y_{t,i} - \delta_i}{1-y_{t,i}}$, and $x_{t+1,i} \gets 1$ w.p. $\frac{\delta_i}{1-y_{t,i}}$
		\Else
		\State $x_{t+1,i} \gets x_{t,i}$ \Comment{Keep the same state}
		\EndIf 
		\State $x_{t+1,i+N} \gets  1 - x_{t+1,i} $ \Comment{Update the augmented states}
		\EndFor
		\State \Return $\x_{t+1}$ \Comment{Return the next physical state satisfying~$\mathbb{E}[\x_{t+1}] = \y_{t+1}$}
		\EndProcedure

	\end{algorithmic}
	\end{footnotesize}
	\caption{Coupled Rounding}
	\label{algo:coupling_scheme}
\end{algorithm}


\subsection{Time Complexity}
\acai{} uses $\OMA$ in Algorithm~\ref{algo:online_mirror_ascent} coupled with a rounding procedure $\DepRound$ or $\Round$. The rounding step may take $\BigO{N}$ operations (amortized every $M$ requests when \DepRound{} is used). In practice, \acai{} quickly sets irrelevant objects in the fractional allocation vector $\y_t$ very close to 0. Therefore, we can  keep track only of objects  with a fractional value above a threshold $\epsilon >0$, and the size of this subset is practically of the order of $h$. 

Similarly,  subgradient computation \nnew{(see \new{\ifnum\extended=1Supplementary material, Sec.~V,~Eq.~(52)\else\cite[Eq.~(63)]{sisalem20techrep}\fi})} may require  $\BigO{N}$ operations per each component and then have  $\BigO{N^2}$ complexity, but in practice, as the vector $\y_t$ is sparse, 
calculations 
require only a constant number of operations and complexity reduces to $\BigO{N}$.

Finally, we use the \Old{negative entropy} negative entropy Bregman projection in~\cite{sisalem21icc}
(line 6 of Algorithm~\ref{algo:online_mirror_ascent}) that has  $\BigO{N  + h \log(h) }$ time complexity. The $\BigO{N  + h \log(h)}$ is due to a partial sorting operation while the actual projection takes $\BigO{h}$. Again, most of the components of $\y_t$ are equal to $0$, so that, in practice, we need to sort much \oold{less} \nnew{fewer} points.


\subsection{Theoretical Guarantees}
\label{sub:guarantees}
The best static cache allocation in hindsight
is the cache state $\x_*$ that maximizes the time-averaged caching gain in Eq.~\eqref{e:gain} over the time horizon $T$, i.e., 
\begin{align}
   \x_*  \in  \underset{\vec x \in \mathcal{X}}{\argmax}\left( G_T (\x) \triangleq\frac{1}{T} \sum^T_{t=1} G(r_t, \x)\right).
   \label{eq:static_gain}
\end{align}
We observe that solving~\eqref{eq:static_gain} is NP-hard in general even for $k=1$ under a stationary request process~\cite{garetto2020similarity}.
Nevertheless,  \acai{} operates in the \emph{online setting} and provides guarantees in terms of the $\psi$-regret~\cite{krause2014submodular}. In this scenario, {the} regret is defined as \nnew{a} gain loss in  comparison to the best static cache allocation $\x_*$ in \eqref{eq:static_gain}.
The $\psi$-regret discounts the \new{best static} gain by a factor $\psi \in (0, 1]$. Formally,
\begin{align}
& \text{$\psi$-}\Regret{\mathcal{X}}{\OMA_{\Phi}} =  \nonumber \\ 
& \underset{\{ r_1, r_2,\dots, r_T\} \in \mathcal{R}^T}{\sup} \left\{ \psi\sum^T_{t=1} G( {r}_t,\x_*) -\mathbb{E} \left[\sum^T_{t=1} G( {r}_t,\x_t)\right] \right\}
    \label{eq:regret},
\end{align}
where the expectation is over the randomized choices of \DepRound.
\new{Note that the supremum in~\eqref{eq:regret} is over all possible request sequences. This definition corresponds to the so\nnew{-}called \emph{adversarial analysis}, imagining that an adversary selects requests in $\mathcal{R}$ to  jeopardize cache performance. \Old{This modeling approach is commonly used to characterize 
system performance under highly volatile external parameters (e.g., the sequence of requests $r_t$) and has been recently successfully applied to caching problems~\cite{paschos2019learning,sisalem21icc,bhattacharjee20}.} \New{ The definition of the regret in Eq.~\eqref{eq:regret}, which compares the gain of the policy to a static offline solution, is classic. Several bandit settings, e.g., simple multi-armed bandits~\cite{radlinski2008learning, kleinberg2008multi, audibert2009exploration}, contextual bandits~\cite{chu2011contextual, agarwal2014taming, dudik2011efficient}, and, of course, their applications to caching problems under the full-information setting~\cite{paschos2019learning, sisalem21icc, 9517925,bhattacharjee20,li2021online}, adopt this definition. In all these cases, the dynamic, adaptive algorithm is compared to a static policy that has full hindsight of the entire trace of actions. Moreover, as is customary \oold{in the context of online problems with  which}
\nnew{when} the offline problem is NP-hard \cite{chen2018online}, the regret is not w.r.t.~the optimal caching gain, but \nnew{w.r.t.}~the gain obtained by an offline approximation algorithm. }
}
\new{\oold{Obviously,} Regret bounds in the adversarial setting provide strong robustness guarantees in practical scenarios.} 
\acai{} has the following regret guarantee:
\begin{theorem}
\label{theorem:main}
Algorithm \ref{algo:online_mirror_ascent} configured with the negative entropy mirror map, learning rate \Old{$\eta=\frac{1}{c_d^k + c_f}\sqrt{\frac{2 \log\left(\frac{N}{h}\right)}{T}}$}\New{$\eta = \frac{1}{(c_d^k + c_f)} \sqrt{\frac{2 \log\left(\frac{N}{h}\right)}{T + (M-1) (M+T)}}$}, and  rounding scheme  \textnormal{\Round{}} or  \textnormal{\DepRound{}} with \nnew{a} freezing period \Old{$ M = 1$} \New{$M=\Theta\left({T^\beta}\right)$ for $\beta \in [0,1)$}, has a sublinear \mbox{$(1-1/e)$-regret} in the number of requests, i.e.,
\Old{\begin{align*}
    \text{$\left(1-1/e\right)$-}& \Regret{\mathcal{X}}{\OMA_{\Phi}} \\
    & \leq \left(1-\frac{1}{e}\right) (c_d^k + c_f) h \sqrt{2 \log{\left(\frac{N}{h}\right)} T},
\end{align*}}
\New{
\begin{align*}
    &\text{$\left(1-1/e\right)$-} \Regret{\mathcal{X}}{\OMA_{\Phi}} \\
    & \leq \left(1-\frac{1}{e}\right) (c_d^k + c_f) h \sqrt{2 \log{\left(\frac{N}{h}\right)}  ( (M-1) (T+M) + T )},
\end{align*}}

where the constant $c_d^k$ is an upper bound on the dissimilarity cost of the k-th closest object  for any request in $\mathcal{R}$.
\label{theorem:knn_total_cost}
\end{theorem}
\old{
A sketch of the proof is available in Supplementary material Sec.~IX (the complete proof is in~\cite{sisalem20techrep}). In the proof we characterize the constant $A$.}
\begin{proof} (sketch)
We first prove that the expected gain of the randomly sampled allocations $\x_t$ is a $(1{-}1/e)$-approximation of the fractional gain. Then, we use online learning results~\cite{bubeck2015convexbook} to bound the regret of OMA schemes operating on a convex decision space against concave gain functions picked by an adversary. The two results are combined to obtain an upper bound on the $(1{-}1/e)$-regret. The full proof is available in Supplementary material\nnew{,} Sec.~IX. 
\end{proof}

\old{
A consequence of Theorem~\ref{theorem:knn_total_cost} is that $\forall \epsilon > 0 $, for $T$ sufficiently large:
\begin{align}
    \mathbb{E} \left[\frac{1}{T}\sum^T_{t=1} G( {r}_t,\x_t)\right] \geq \left(1-\frac{1}{e} -\epsilon \right) \frac{1}{T}\sum^T_{t=1} G( {r}_t,\x_*).
\end{align}
Then \acai{} performs on average as well as a $(1-1/e)$-approximation of the optimal configuration $\x_*$ (NP-hard to calculate). 
}
\new{
\New{The $(1-1/e)$-regret of \acai{} under \Round{} scheme \oold{with its corresponding freezing period $M=1$} has order-optimal regret $\BigO{\sqrt{T}}$~\cite{hazan2016introduction}. \oold{whereas} \nnew{Under} the rounding scheme \DepRound{} with \oold{a corresponding} freezing period $M =\Theta\left({T^\beta}\right)$, the reduced reactivity of \acai{} is reflected by the additional $T^\frac{\beta}{2}$ factor in the order of the regret.}
\Old{A consequence of Theorem~\ref{theorem:knn_total_cost} is that  } Nonetheless, the expected time-average $(1-1/e)$-regret of \acai{} can get arbitrarily close to zero for \nnew{a} large time horizon. Hence, \acai{} performs on average as well as a $(1-1/e)$-approximation of the optimal configuration $\x_*$.   This observation also suggests that our algorithm can be used as an iterative method to solve the NP-hard static allocation problem with the best approximation bound 
achievable for \oold{this kind} \nnew{these kinds} of problem{s}~\cite{nemhauser1978best}. 
}
\begin{corollary}
\label{corollary:offline}
\new{(offline solution) 
Let $\bar{\y}$ be the average fractional allocation $\bar{\y} =\frac{1}{\tilde T} \sum^{\tilde T}_{i=1} \y_i$ of \textnormal{\acai{}}, and $\bar{\x}$ the random state sampled from $\bar{\y}$ through \textnormal{\Round} or \textnormal{\DepRound{}}. If Algorithm~\ref{algo:online_mirror_ascent} is configured with the negative entropy mirror map, and, at each iteration $t \in [\tilde T]$, operates with subgradients of the time-averaged caching gain~\eqref{eq:static_gain},
then $\forall \epsilon > 0$ and over a sufficiently large number of iterations $\tilde T$, $\bar{\x}$ satisfies}
\begin{align}
\nonumber
       \mathbb{E} \left[ G_T(\bar{\x})\right]\geq \left(1-\frac{1}{e} - \epsilon\right) G_T({\x_*}).
\end{align}
where $\x_* =  \underset{\x \in \mathcal{X}}{\arg\max}~G_T(\bar{\x})$.
\end{corollary}
The proof can be found in Supplementary material, Sec.~X.

\section{Experiments}
\label{sec:experiments}

We start evaluating \acai{} in a simple scenario with a synthetic request process, for which we can compute the  optimal fractional static cache allocation. We then consider real-world catalogs and traces and  compare our solution  with state\nnew{-}of\nnew{-}the\nnew{-}art online policies proposed for \knntext{} caching, i.e.,  \simlru{} \cite{pandey2009nearest}, \clslru{} \cite{pandey2009nearest}, and \qcache{} \cite{FALCHI2012803} described in Sec.~\ref{sec:bg_similarity}.

\subsection{Simple Scenario}
As in~\cite{sabnis2021grades}, we consider a synthetic catalog 
 of $N=900$ objects positioned on a $30\times30$ grid. 
The request process is generated according to the Independent Reference Model~\cite{coffman1973operating}. The objects' popularity is represented by a Gaussian distribution. In particular, an object $o \in \mathcal{N}$  with \oold{an} $l_1$ \oold{(norm-1)} distance $d_o$ from the center of the grid (15, 15)  is requested  at any time slot $t$ with probability $p_o \propto e^{- \frac{d^2_o}{2 \times 6^2}}$. The synthetic catalog is depicted in Fig.~\ref{fig:synthetic_catalog}.

We consider the dissimilarity cost to be the $l_1$ distance. We take different values for the retrieval cost $c_f \in \{1, 2,3,4\}$ and the number of neighbors  $k \in \{1,2,3,4,5\}$. We take the cache capacity to be $h = 15$. We configure $\acai$ with \DepRound{} rounding scheme.

\begin{figure}[!t]
    \centering
    \includegraphics[width=.5\linewidth]{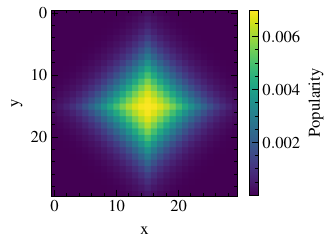}\vspace{-.5em}
    \caption{Synthetic catalog of objects located on a $30\times30$ grid. The heatmap depicts the popularity distribution of objects in the grid. 
    }
    \label{fig:synthetic_catalog}
\end{figure}

We use CVXPY~\cite{diamond2016cvxpy} to find the optimal fractional static cache allocation, and \acai{} to compute its approximation according to  Corollary~\ref{corollary:offline}. In particular, \acai{} runs for $T = 10\,000$ iterations with a diminishing learning rate $\eta_t = \frac{2.0}{c_f} (1 + \cos(\frac{\pi t}{T}))$.


\begin{figure}[!t]
    \centering
    \subcaptionbox{Optimal fractional static cache allocations}{\includegraphics[width=.95\linewidth]{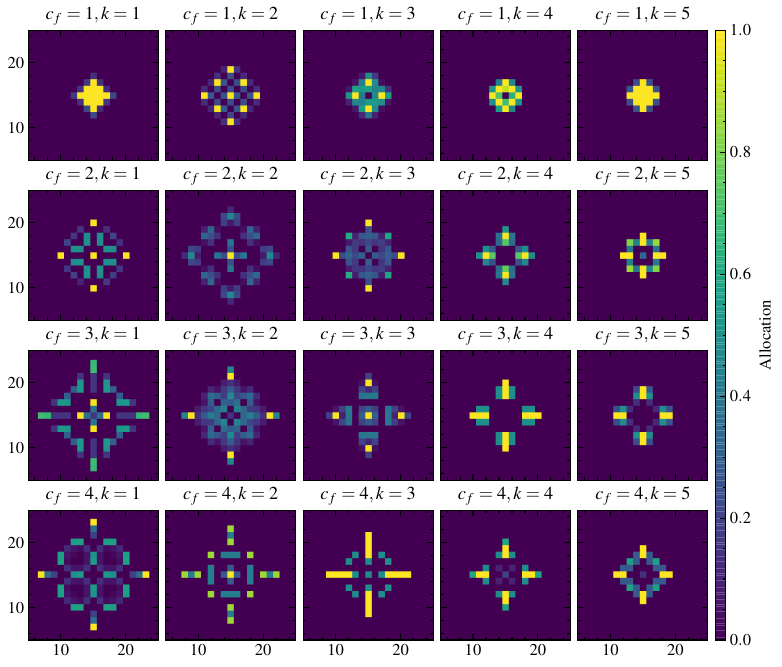}}
    \subcaptionbox{\acai's fractional static cache allocations}{\includegraphics[width=.95\linewidth]{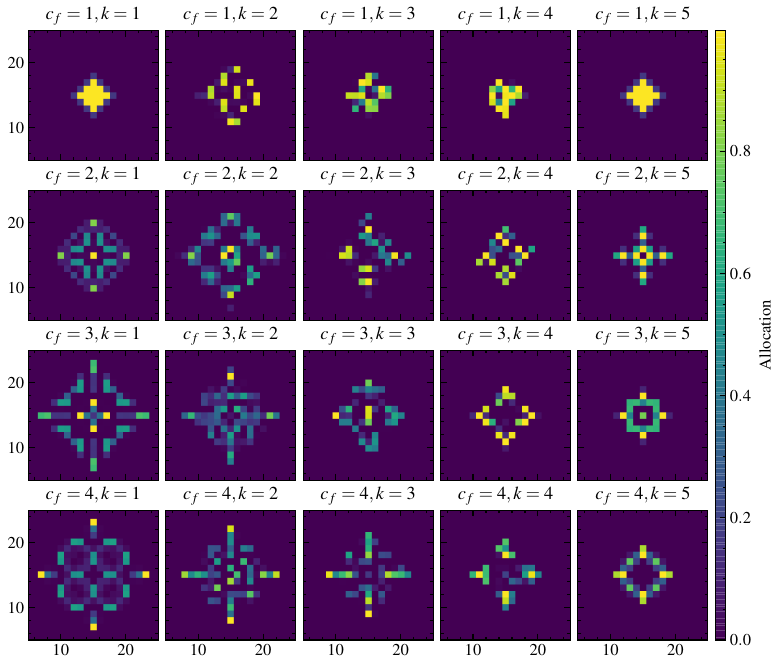}}
    \caption{The optimal fractional static cache allocations, and \acai's fractional static cache allocations under different values of \nnew{the} retrieval costs $c_f \in \{1, 2,3,4\}$ and \nnew{of the} number of neighbors $k \in \{1,2,3,4,5\}$.  }
    \label{fig:allocations}
\end{figure}

\vspace{2mm}
\noindent \textbf{Results.} 
Figure~\ref{fig:allocations} shows the optimal fractional static cache allocations and the \acai's fractional static cache allocations (see Corollary~\ref{corollary:offline}) under different retrieval costs and the number of neighbors~$k$. We observe that the optimal fractional static cache allocations in Fig.~\ref{fig:allocations}~(a) are symmetric, while \acai's fractional static cache allocations in  Fig.~\ref{fig:allocations}~(b) partially lose this symmetry  for values of $k \in \{2,3,4\}$ primarily due to the different ways a \knntext{} query can be satisfied over the physical catalog for such values. In fact, 
there are multiple objects in the catalog with the same distance from a request $r$, and \acai{} only selects a single permutation $\pi^r$ for a request~$r$. We observe that, when the retrieval cost is higher, the allocations are more spread to cover a larger part of the popular region. \oold{Moreover, for larger values of the number~$k$ of objects to be served,  {more mass is added to the neighborhood of}  the cached objects, but when the number~$k$ of objects to be served changes from $k = 1$ to $k=2$ it is not clearly observed. }
\begin{figure}[!t]
    \centering
    \includegraphics[width=.6\linewidth]{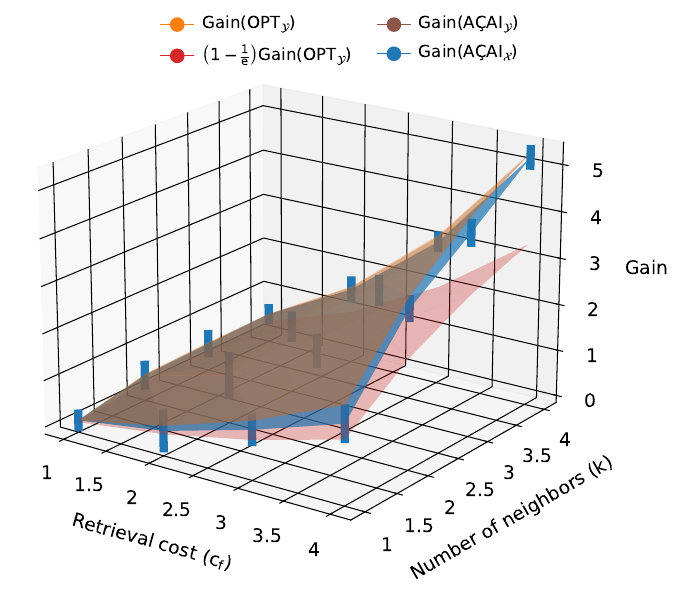}
    \caption{The gain of the optimal fractional static cache allocations \nnew{($\mathrm{Gain} (\mathrm{OPT}_{\mathcal Y})$)}, its $(1-\frac{1}{e})$-approximation, and the gain obtained by \acai's fractional static \oold{approximated} cache allocations \nnew{($\mathrm{Gain} (\mathrm{A\c{C}AI}_{\mathcal Y})$)} and static integral \oold{approximated} cache allocations \nnew{($\mathrm{Gain} (\mathrm{A\c{C}AI}_{\mathcal X})$)} under different values of \nnew{the} retrieval cost $c_f \in \{1, 2,3,4\}$ and \nnew{of the} number of neighbors $k \in \{1,2,3,4,5\}$. For \acai{}'s static integral \oold{approximated} cache allocations, we report 
    95\% confidence intervals computed over 50 different runs. The gain of the fractional static \oold{approximated} cache allocations obtained by \acai{} overlaps with \nnew{the} gain of the optimal fractional static {cache} allocations.}
    \label{fig:static_vs_fractional_vs_integral}
\end{figure}

In Fig.~\ref{fig:static_vs_fractional_vs_integral}, we compare the gain obtained by the optimal fractional static {cache} \nnew{allocations}, its $(1-\frac{1}{e})$-approximation, and the gain obtained by \nnew{\acai's} fractional static \oold{approximated} \nnew{{cache}} allocation{s}  and the static integral \oold{approximated} \nnew{{cache}} allocation{s} \oold{of \acai~(obtained through \DepRound)} \nnew{(see Corollary~\ref{corollary:offline})}. 
While Fig.~\ref{fig:allocations} shows that \acai{}'s fractional static cache allocations may differ from the optimal ones, their costs are practically indistinguishable (the two corresponding surfaces overlap in Fig.~\ref{fig:static_vs_fractional_vs_integral}). 
We also observe  that rounding comes at a cost, as there is a clear gap between the gain of \acai{}'s fractional \nnew{{cache}} allocations  and of \acai{}'s integral ones. Still, the average gain of the integral \nnew{{cache}} allocations remains close to the fractional optimum and well above the $(1-\frac{1}{e})$-approximation. \acai{} then performs much better than what \nnew{is} guaranteed by Corollary~\ref{corollary:offline} (a potential gain reduction by a factor $1-\frac{1}{e}$).

\subsection{Real-world Datasets}

\noindent \textbf{SIFT1M trace.} SIFT1M is a classic benchmark \oold{data-set} \nnew{dataset} to evaluate approximate \knntext{} algorithms~\cite{jegou2010product}. It contains 1 million objects embedded as points in a 128-dimensional space. SIFT1M does not provide a request trace so we generated a synthetic one according to the Independent Reference Model~\cite{coffman1973operating} (similar to what \nnew{is} done in other papers like~\cite{falchi2008metric, guo2018potluck}). Request $r_t$ is for object $i$ with a  probability $\lambda_i$ independently \oold {from} \nnew {of} previous requests.
We spatially correlated \oold{requests} \nnew{objects' popularities} by letting $\lambda_i$ depend on the position of the embeddings in the space. In particular, we considered the barycenter of the whole dataset and set $\lambda_i$ proportional to $d_i^{-\beta}$, where $d_i$ is the distance of~$i$ from the barycenter. The parameter $\beta$ was chosen such that the tail of the ranked objects\nnew{'} popularity distribution is similar to a Zipf with parameter $0.9$, as observed in some image retrieval systems~\cite{FALCHI2012803}. 
We generated a trace with $10^5$ requests. 
The number of distinct objects requested in the trace is approximately $2 \times 10^4$.

\vspace{2mm}
\noindent \textbf{Amazon trace.} The authors of \cite{mcauley2015image} crawled the Amazon web-store and collected a dataset to model relationships among products and provide user recommendations. They took as input the visual features of product images obtained from a machine learning model pre-trained on 1.2 million images from ImageNet. \New{The visual features are augmented with the relationships between the items, and these relationships are collected based on the cosine similarity of the sets of users who purchased or viewed the items. The objects' dissimilarity is modeled as a distance $d(\cdot,\cdot)$, such that $\mathbb P(\text{item }i \text{ is related to item } j)$ increases monotonically with $d(\x_i ,\x_j)$, where $\x_i$ and $\x_j$ are the visual features of the items $i$ and $j$. The authors of~\cite{mcauley2015image} show that the relationship `users who viewed $i$ also viewed $j$' can successfully be used to provide accurate recommendations.}
The authors of \cite{sabnis21} built a request trace from the timestamped user reviews for objects in the category Baby embedded in a $100$-dimensional space. Two products $o$ and $o'$ are considered similar if they have been viewed by the same users.  We use the request trace from~\cite{sabnis21} and in particular the interval $[2 \times 10^5,3 \times 10^5]$.\footnote{\new{We discard the initial part of the trace because it contains requests only for a small set of objects
(likely the set of products to crawl was progressively extended during the measurement campaign in~\cite{mcauley2015image}).
} } The number of distinct objects requested in this trace is approximately $2 \times 10^4$.

\subsection{Settings and Performance Metrics}

For \acai, unless otherwise said, we choose the negative entropy $\Phi(\y) = \sum_{i \in \mathcal{N}} y_i \log (y_i)$ as mirror map (see Fig.~\ref{fig:negentropy_vs_euclidean} and the corresponding discussion for other choices)  and the rounding scheme \DepRound{} with $M=1$.
The learning rate is set to the best value found exploring the range $[10^{-6}, 10^{-4}]$.


As for the state-of-the-art caching policies, \simlru{} and \clslru{} have two parameters, $C_\theta$ and $k'$, that we set in each experiment to the best values we found exploring the ranges $[c_f, 2c_f]$ for $C_\theta$ and $[1, h]$ for $k'$. For \qcache{} we consider $l=h/k$: the cache can then perform the \knntext{} search over all local objects.

We also consider a simple similarity caching policy that stores previous requests and the corresponding set of $k$ closest objects as key-value pairs, and manages the set of keys according to \lru. The cache then serves locally the request if it coincides with one of the previous requests in its memory, it forwards it to the server, otherwise. The ordered list of keys is updated as in \lru. We refer to this policy simply as \lru.
We  compare the policies in terms of their normalized average caching gain \oold{per-request} \nnew{per request}, where the normalization factor corresponds to the caching gain of a cache with \nnew{a} size equal to the whole catalog. In such \nnew{a} case, the cache could store the entire catalog locally and would achieve the same dissimilarity cost of the server without paying any fetching cost. The maximum possible caching gain is then $k c_f$. The normalized average gain of a policy $\mathcal{P}$ with cache states $\{\x_{t}\}^T_{t=1}$ over $T$ requests can then be defined as:
\begin{align}
   \mathrm{NAG} (\mathcal{P}) = \frac{1}{k c_f T} \sum^T_{t=1} G(r_t,\x_t).
\end{align}

\subsection{Results}

We consider a dissimilarity cost proportional to the squared Euclidean distance. This is the usual metric considered for SIFT1M benchmark and also the one  considered to learn the embeddings for the Amazon trace in~\cite{mcauley2015image}.

The numerical value of the fetching cost depends on its interpretation (delay experienced by the user, load on the server or on the network) as well as on the application, because it needs to be converted into the same unit of the approximation cost. In our evaluation, we let it depend on the topological characteristics of the dataset in order to be able to compare the results for the two different traces. Unless otherwise said, we set $c_f$ equal to the average distance of the 50-th closest neighbor \new{in the catalog $\mathcal N$}.

\begin{figure}[!t]
\captionsetup[subfigure]{aboveskip=0pt,belowskip=-1pt}
    \centering
      \subcaptionbox{SIFT1M  trace}{
     \includegraphics[width=.46\linewidth]{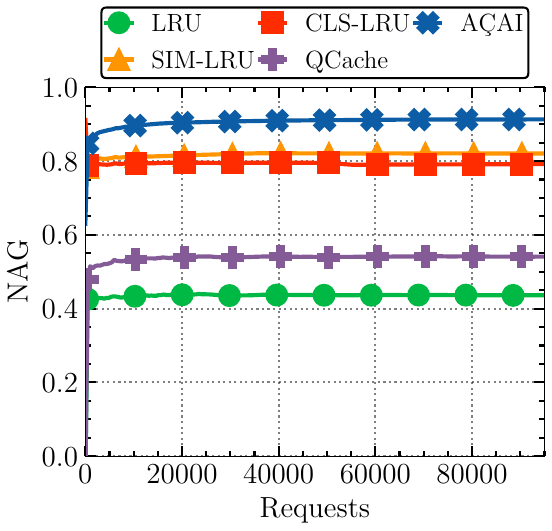}
    }
    \subcaptionbox{Amazon trace}{
    \includegraphics[width=.46\linewidth]{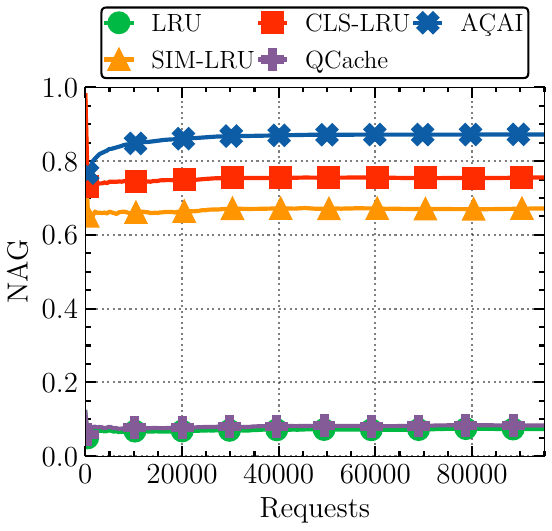}
   }
    \caption{Caching gain for the different policies. The cache size is $h=1000$ and $k = 10$.}
    \label{fig:dynamic_behaviour}
\end{figure}

Figure~\ref{fig:dynamic_behaviour} shows how the normalized average gain changes  as requests arrive and the different caching policies update the local set of objects (starting from an empty configuration). The cache size is $h=1000$ and the cache provides $k=10$ similar objects for each request. All policies reach an almost stationary gain after at most a few thousand requests. Unsurprisingly, the na\"ive \lru{} has the lowest gain (it can only satisfy locally requests that match exactly a previous request) and similarity caching policies perform better. \acai{} has a significant improvement in comparison to the second best policy (\simlru{} for SIFT1M and \clslru{} for Amazon). 

\begin{figure}[!t]
\captionsetup[subfigure]{aboveskip=-1pt,belowskip=-1pt}
    \centering
        \subcaptionbox{SIFT1M  trace}{
         \includegraphics[width=.46\linewidth]{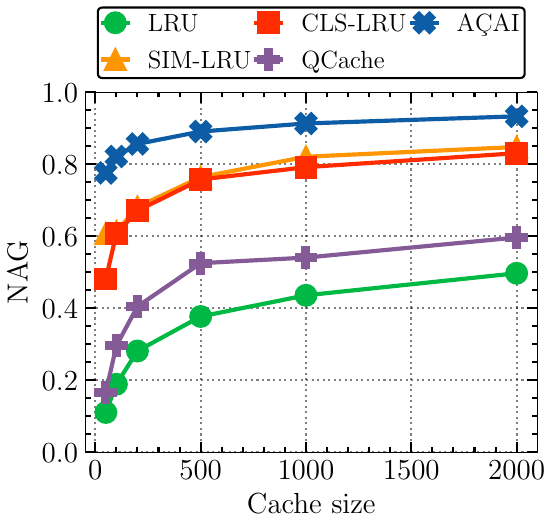}
        }
    \subcaptionbox{Amazon trace}{
     \includegraphics[width=.46\linewidth]{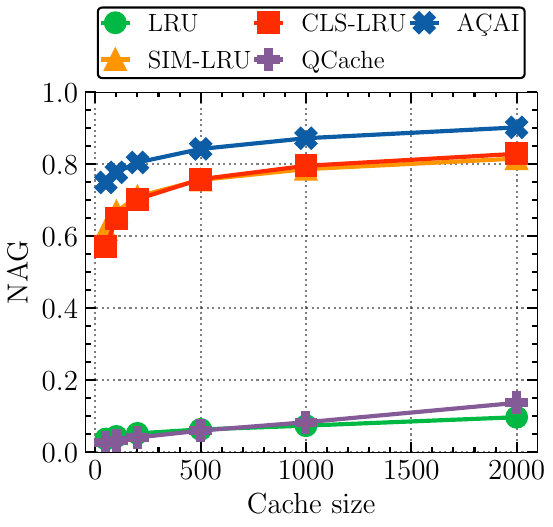}
    }

    \caption{Caching gain for the different policies, for different cache sizes $h \in \{50, 100, 200,500, 1000, 2000\}$ and $k =10$.  }
    \label{fig:gain_vs_cache_size}
\end{figure}

\begin{figure}[!t]
\captionsetup[subfigure]{aboveskip=-1pt,belowskip=-1pt}
    \centering
    \subcaptionbox{SIFT1M trace}{
    \includegraphics[width=.46\linewidth]{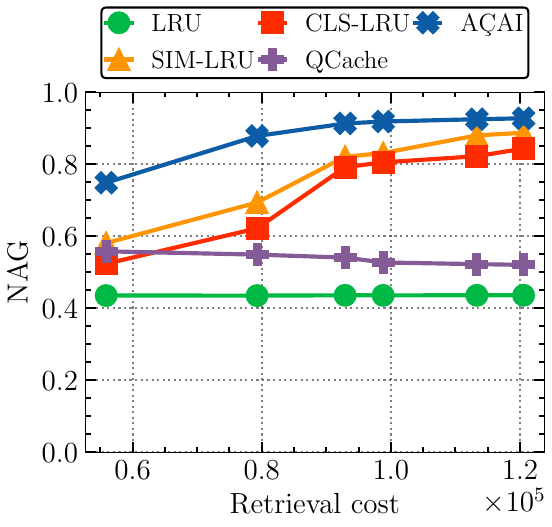}
    }
    \subcaptionbox{Amazon trace}{
    \includegraphics[width=.46\linewidth]{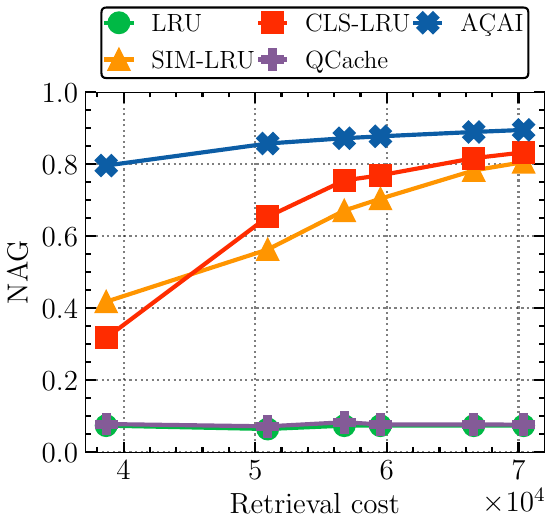}
    }
    \caption{Caching gain for the different policies and different retrieval cost\nnew{s}. The retrieval cost $c_f$ is taken as the average distance to the $i$-th neighbor, $i \in \{2, 10, 50, 100, 500, 1000\}$. The cache size is $h= 1000$ and $k = 10$. }
    \label{fig:gain_vs_fetching_cost}
\end{figure}

\begin{figure}[!t]
\captionsetup[subfigure]{aboveskip=-1pt,belowskip=-1pt}
    \centering
        \subcaptionbox{SIFT1M  trace}{
    \includegraphics[width=.46\linewidth]{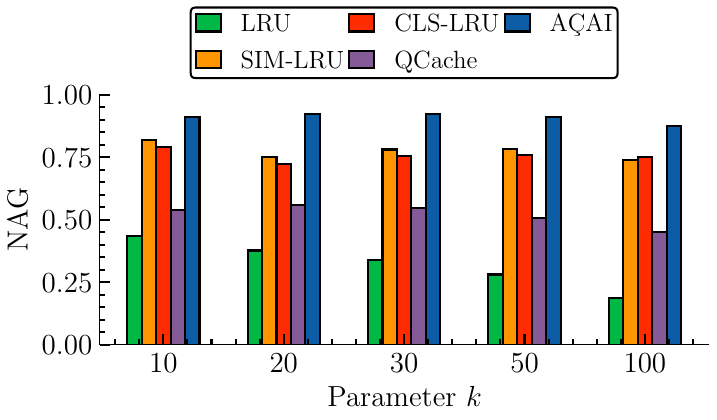}
    }
    \subcaptionbox{Amazon trace}{
     \includegraphics[width=.46\linewidth]{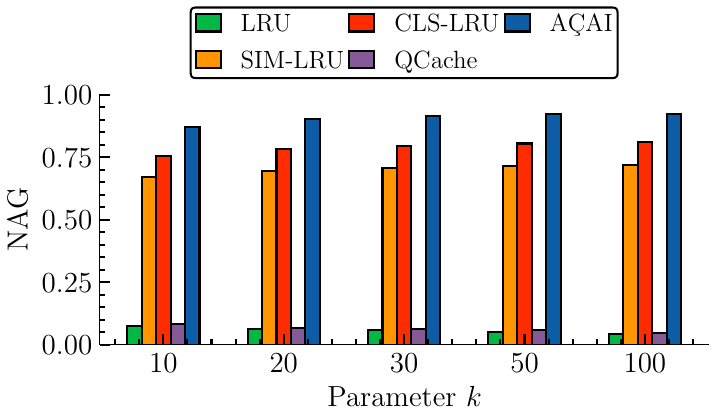}
    }
  \caption{Caching gain for the different policies. The cache size is $h = 1000$, and $k \in \{10,20,30,50,100\}$. }
    \label{fig:gain_vs_k}
\end{figure}

This advantage of \acai{} is constantly confirmed for different cache sizes (Fig.~\ref{fig:gain_vs_cache_size}), different values of the fetching cost~$c_f$ (Fig.~\ref{fig:gain_vs_fetching_cost}), and different values of $k$ (Fig.~\ref{fig:gain_vs_k}). The relative improvement of \acai{}, in comparison to the second best policy, is larger for small values of the cache size ($+30\%$ for SIFT1M and $+25\%$  for Amazon when $h=50$), and small values of the fetching cost ($+35\%$ for SIFT1M and $+100\%$ for $c_f$ equal to the average distance from the second closest object). Note how these are the settings where caching choices are more difficult (and indeed all policies have lower gains): when  cache storage can accommodate only a few objects, it is critical to carefully select which ones to store; when the server is close, the costs of serving requests from the cache or \oold{from} the server are similar and it is difficult to correctly decide how to satisfy the request. 
\nnew{The performance of} caching policies \oold{performance} \oold{are} \nnew{is} in general less dependent on the number $k$ of similar objects to retrieve and \acai{} achieves about $10\%$ improvement for $k$ between $10$ and $100$ when $h=1000$ (Fig.~\ref{fig:gain_vs_k}).

\begin{figure}[!t]

    \centering
         \subcaptionbox{\acai, $h = 50$}{
       \includegraphics[width=.46\linewidth]{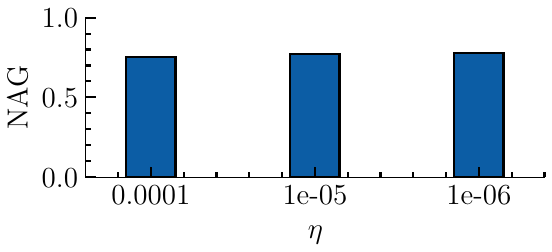}
    } \subcaptionbox{\acai, $h = 10^3$}{
       \includegraphics[width=.45\linewidth]{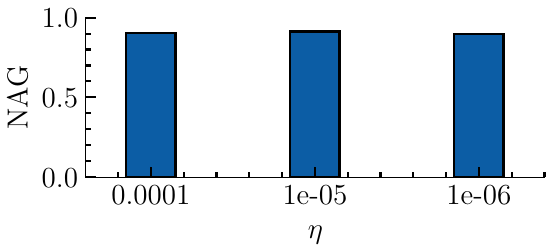}
    }

      \subcaptionbox{\simlru, $h = 50$}{
    \includegraphics[width=.46\linewidth]{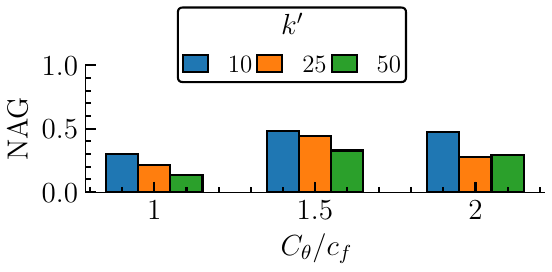}
    
    }
           \subcaptionbox{\simlru, $h = 10^3$}{
     \includegraphics[width=.46\linewidth]{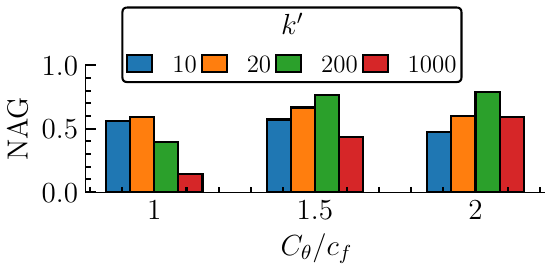}
    }
    
    \subcaptionbox{\clslru, $h = 50$}{
    \includegraphics[width=.46\linewidth]{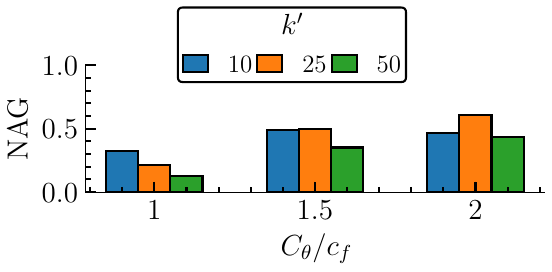}
    }
        \subcaptionbox{\clslru, $h = 10^3$}{
     \includegraphics[width=.46\linewidth]{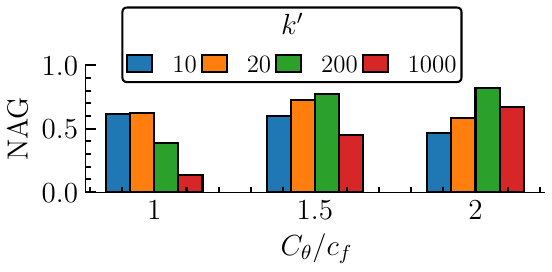}
    }

    \caption{Caching gain for \acai{} for different values of $\eta$ (top). Caching gain for \simlru (middle) and \clslru{} (bottom) for different values of the parameters $(k', C_\theta)$. SIFT1M  trace. }
    \label{fig:parameter_configuration}
\end{figure}

\vspace{2mm}
\noindent
{\bf Sensitivity analysis.}
We now evaluate the robustness of \acai{} to the configuration of its single parameter (the learning rate~$\eta$). Figure~\ref{fig:parameter_configuration} shows indeed that, for learning rates that are two orders of magnitude apart, we can achieve almost the same normalized average gain both for $h=50$  and for $h=1000$.\footnote{
    \new{Under a stationary request process, a smaller learning rate would lead to converge slower but to a solution closer to the optimal one. Under a non-stationary process, a higher learning rate may allow faster adaptivity. In this trace, the two effects almost compensate, but see also Fig.~\ref{fig:negentropy_vs_euclidean}.}
} 

In contrast, the performance of the second best policies (\simlru{} and \clslru) are more sensitive to the choice of their two configuration parameters $k'$ and $C_\theta$. For example, the optimal configuration of \simlru{} is $k'=10$ and $C_\theta = 1.5 \times c_f$ for a small cache ($h=50$) but $k'=200$ and $C_\theta = 2 \times c_f$ for a large one ($h=1000$). Moreover, in both cases a misconfiguration of these parameters would lead to significant performance degradation.

\begin{figure}[!t]
    \centering
  \setlength\abovecaptionskip{-0.3\baselineskip}
    \includegraphics[width=.6\linewidth]{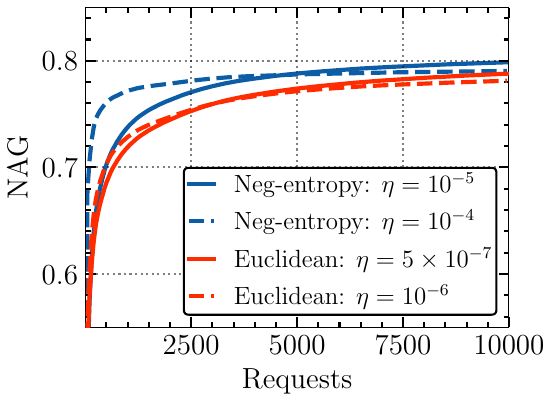}
    \caption{Caching gain for \acai{} configured with negative entropy and Euclidean maps (SIFT1M trace). The cache size is \new{$h=100$} and $ k= 10$.}
    \label{fig:negentropy_vs_euclidean}
\end{figure}

\vspace{2mm}
\noindent
{\bf Choice of the mirror map.}
If the mirror map is selected equal to the squared Euclidean norm, the $\OMA$ update coincides with a standard gradient update. Figure~\ref{fig:negentropy_vs_euclidean} shows the superiority of the negative entropy map: \new{it allows \nnew{one} to achieve a higher gain than the Euclidean norm map or the same gain but in a shorter time}. To the best of our knowledge, ours is the first paper that shows the advantage of using non-Euclidean mirror maps for \New{similarity} caching problems. 
\nnew{We observe how our finding is in apparent contrast with what observed  for exact caches in~\cite{sisalem21icc}, i.e., that the Euclidean mirror map should always be preferred when requests are not batched. The  difference can be explained as follows: under exact caching only the requested object can satisfy the request and then the gradient has a single non-null component, but, under similarity caching, multiple objects in the vicinity of a request can contribute to reduce the cost of serving it and the gradient is then less sparse. It is known that denser gradients may lead to prefer the negative entropy map \cite[Sec.~4.3]{bubeck2015convexbook} 
and our results in \ifnum\extended=1 Supplementary material\nnew{,} Sec.~VII-A\else \cite[Appendix D]{sisalem20techrep}\fi~provide a theoretical justification for our specific problem. 
}

\vspace{2mm}
\noindent
{\bf Dissecting \acai{} performance.}
In comparison to state-of-the-art similarity  caching  policies, \acai{} introduces two key ingredients: (i) the use of fast, approximate indexes to decide what to serve from the local catalog and what from the remote one, and (ii) the $\OMA$ algorithm to update the cache state. It is useful to understand how much each ingredient contributes to \acai{} improvement with respect to the other policies.  

To this aim, we integrated the same indexes in the other policies allowing them to serve requests as \acai{} does, combining  both  local  objects  and  remote ones \oold{on the basis of} \nnew{based on} their costs (see Sec.~\ref{sec:request_serving}), while leaving their cache updating mechanism unchanged. We then compute, in the same setting of Fig.~\ref{fig:gain_vs_k}, how much the gain of the second best policy (\simlru{} for SIFT1M and \clslru{} for Amazon) increases because of \acai{} request service mechanism. This is the part of \acai{} improvement attributed to the use of the two indexes, the rest is attributed to the cache update mechanism through $\OMA$. We observe from Fig.~\ref{fig:caching_gain_contribution} that most of \acai{} gain improvement over the second best caching policy is due to the use of approximate indexes, but $\OMA$ updates are still responsible for $15$--$20\%$ of \acai{} performance improvement  under SIFT1M trace and \oold{for} $20$--$35\%$ for the Amazon trace. 

\begin{figure}[!t]
\captionsetup[subfigure]{aboveskip=-1pt,belowskip=-1pt}
    \centering
        \subcaptionbox{SIFT1M  trace}{
    \includegraphics[width=.46\linewidth]{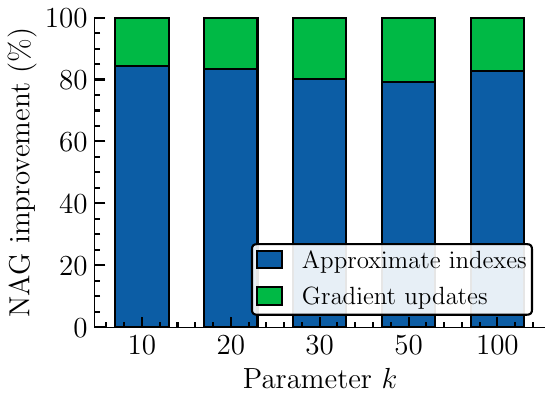}
    }
    \subcaptionbox{Amazon trace}{
    \includegraphics[width=.46\linewidth]{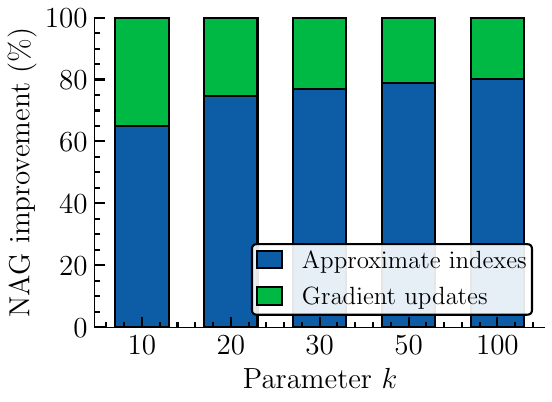}
    }

  \caption{\acai{} caching gain improvement in comparison to the second best state-of-the-art similarity caching policy: contribution of approximate indexes and gradient updates. The cache size is $h = 1000$, and $k \in \{10,20,30,50,100\}$. }
    \label{fig:caching_gain_contribution}
\end{figure}

\vspace{2mm}
\noindent
{\bf Update cost.}
In this part, we evaluate the update cost of the different rounding schemes. We set the cache size $h= 1000$ and $k = 10$. We run \acai{} over the {Amazon trace} with a learning rate $\eta = 10^{-5}$.

Figure~\ref{fig:amazon_update_cost} (a) gives the time-averaged number of files fetched and  Figure~\ref{fig:amazon_update_cost} (b) gives the caching gain of the different rounding schemes. We observe \nnew{that, }by increasing the cache state freezing parameter $M$, the system fetches \oold{less} \nnew{fewer} files per iteration at the expense of losing reactivity \nnew{and then incurring a smaller gain} at the start. The coupled rounding scheme achieves the best performance as it fetches \oold{less} \nnew{fewer} files without losing reactivity. 

\begin{figure}[t]
\centering
  \subcaptionbox{Time-averaged fetched files}{
    \includegraphics[width=.46\linewidth]{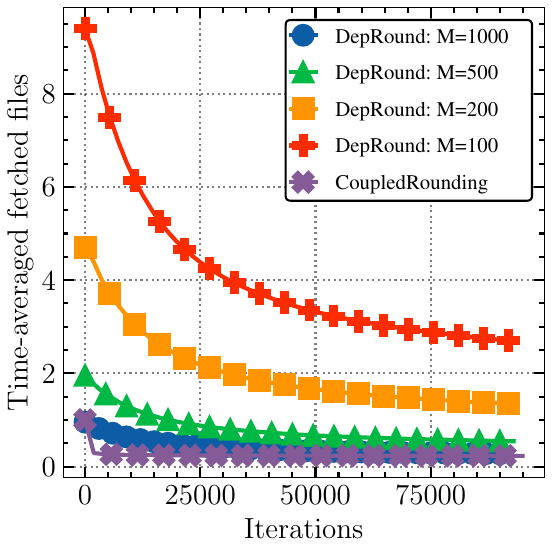}
  } 
  \subcaptionbox{Caching gain}{
 
    \includegraphics[width=.46\linewidth]{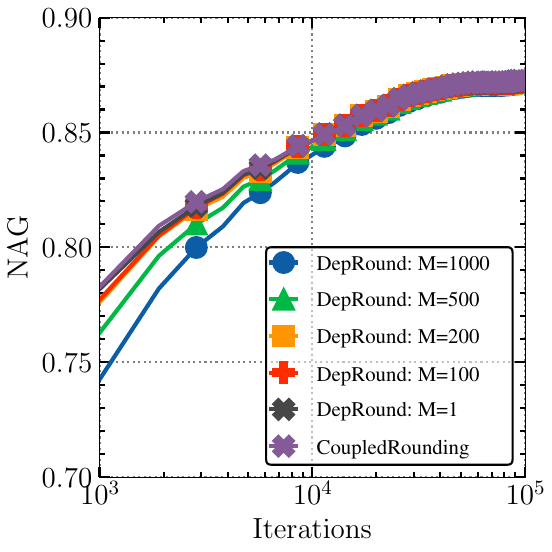}
  }
  \caption{Number of fetched files (time average) and caching gain of \acai{} under different rounding schemes. \acai{} is run with the learning rate $\eta = 10^{-5}$ over the {Amazon trace}. The cache size is $h= 1000$ and $k = 10$.}
\label{fig:amazon_update_cost}
\end{figure}
Figure~\ref{fig:amazon_cache_occupancy} \nnew{shows} \oold{gives} the instantaneous and time\nnew{-}averaged cache occupancy of the cache using the \Round~scheme with the relaxed capacity constraint. We observe that the time\nnew{-}averaged cache occupancy rapidly converges to the cache capacity $h$, while the instantaneous occupancy is kept within 5\% of the cache capacity.
\begin{figure}[t]
    \centering
    \includegraphics[width=.46\linewidth]{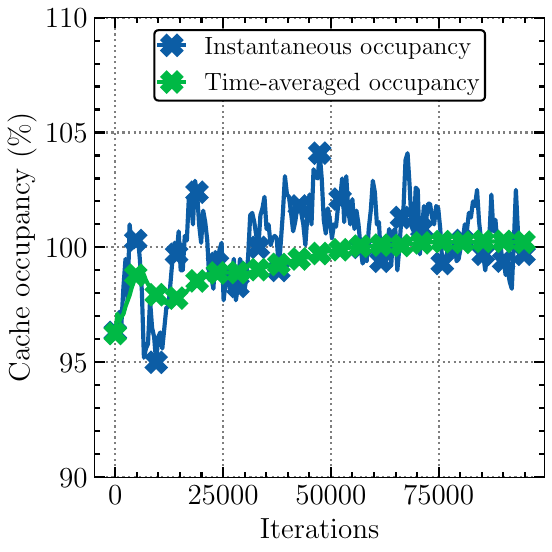}
    \caption{Time-averaged and instantaneous cache occupancy under \Round. \acai{} is run with the learning rate $\eta = 10^{-5}$ over the {Amazon trace}. The cache size is $h= 1000$ and $k = 10$. }
    \label{fig:amazon_cache_occupancy}
\end{figure}
 \section{Conclusion} 
\label{Conclusion}
Edge computing provides computing and storage resources that may enable complex applications with tight delay guarantee\nnew{s} like augmented-reality ones, but these strategically positioned resources need to be used efficiently. To this aim, we designed \acai{}, a content cache management policy that determines dynamically the best content to store on the edge server. Our solution adapts to the user requests, without any assumption on the traffic arrival pattern. \acai{} leverages two key components: (i)~new efficient content indexing methods to keep track of both local and remote content, and (ii) mirror ascending techniques to optimally select the content to store. The results show that \acai{} is able to outperform the state-of-the-art policies and does not need careful parameter tuning.

As future work, we plan to evaluate \acai{} in the context of machine learning classification tasks~\cite{khandelwal2019generalization}, in which the size of the objects in the catalog is comparable to  their $d$-dimensional representation in the index, and, as a consequence, the index size cannot be neglected in comparison to the local catalog size. \New{Another important future research direction is to consider dynamic regret, whereby the performance of a policy is compared to a dynamic optimum. Also, since the employed online algorithm OMA is greedy (i.e., does not keep track of the history of the requests), with careful selection of the mirror map, it may have an adaptive regret guarantee, e.g., such guarantee holds for OGD~\cite{Zinkevich2003}.}

\bibliographystyle{IEEEtran}
\bibliography{0.bibliography.bib}

\ifnum\extended=1



\newpage
\onecolumn
\begin{center}
{\Large  Supplementary Material for Paper:}\\
\vspace{0.4 cm}
{\huge Ascent Similarity Caching with Approximate Indexes}
\end{center}

\setcounter{section}{0}

\section{Equivalent Expression of the Cost Function}

\begin{lemma}
\label{lemma:mins_subtraction}
{Let us fix the threshold} $c \in \mathbb{N} \cup \{0\}$, $r \in \mathcal{R}$ and $i \in \mathcal{U}$. The following equality holds
\begin{align}
   & \min\left\{c, \sum^{i}_{j=1} x_{\pi^r_j}\right\} - \min\left\{c, \sum^{i-1}_{j=1} x_{\pi^r_j}\right\} =x_{\pi_i^r} \mathds{1}_{\left\{ \sum^{i-1}_{j=1} x_{\pi_j^r} < c\right\}} \label{eq:mins_subtraction}.
\end{align}
\end{lemma}
\begin{proof}
We distinguish two cases:
\begin{enumerate}[label=(\roman*)]
\item When $\sum^{i-1}_{j=1} x_{\pi^r_j}\geq c$  this implies that $ \sum^{i}_{j=1} x_{\pi^r_j} =  x_{\pi_j^r} + \sum^{i-1}_{j=1} x_{\pi^r_j} \geq c + x_{\pi^r_i} \geq c$ since $x_{\pi^r_i} \geq 0$. Therefore, $\sum^{i-1}_{j=1} x_{\pi^r_j}\geq c$ implies that $\min\left\{c, \sum^{i-1}_{j=1} x_{\pi^r_j}\right\} = \min\left\{c, \sum^{i}_{j=1} x_{\pi^r_j}\right\} = c$, and we have:
\begin{align}
      \min\left\{c, \sum^{i}_{j=1} x_{\pi^r_j}\right\} - \min\left\{c, \sum^{i-1}_{j=1} x_{\pi^r_j}\right\} &= 0.
       \label{eq:mins_part1}
\end{align}
\item 
When $\sum^{i-1}_{j=1} x_{\pi^r_j} < c$, we have $ \sum^{i}_{j=1} x_{\pi^r_j} =  x_{\pi_j^r} + \sum^{i-1}_{j=1} x_{\pi^r_j} < c + x_{\pi^r_i} \leq c$ since $x_{\pi^r_i} \leq 1$, this implies that $ \min\left\{c, \sum^{i}_{j=1} x_{\pi^r_j}\right\} = \sum^{i}_{j=1} x_{\pi^r_j}$ and $ \min\left\{c, \sum^{i-1}_{j=1} x_{\pi^r_j}\right\} = \sum^{i-1}_{j=1} x_{\pi^r_j}$, and we have
\begin{align}
      \min\left\{c, \sum^{i}_{j=1} x_{\pi^r_j}\right\} - \min\left\{c, \sum^{i-1}_{j=1} x_{\pi^r_j}\right\} &= x_{\pi^r_i}. \label{eq:mins_part2}
\end{align}

\end{enumerate}
Combining Eqs.~\eqref{eq:mins_part1} and \eqref{eq:mins_part2} yields Eq.~\eqref{eq:mins_subtraction}.
\end{proof}

\begin{lemma}
The cost function $C(r, \x)$ given by the expression in Eq.~\eqref{eq:Ca_cost_function}, can be equivalently expressed as:
\begin{align}
    C(r, \x) &=  -\sum^{ K^r-1}_{i=1} \alpha^r_i \min\left\{k - \sigma^r_i, \sum^i_{j=1} x_{\pi^r_j} - \sigma^r_i \right\} + \sum^{K^r}_{i=1}  c(r,\pi^r_{i}) \mathds{1}_{\left\{ \pi^r_i \in \mathcal{U} \setminus \mathcal{N}\right\}},
    \label{eq:Ca_expansion_simple}
\end{align}
where $\sigma^r_i =  \sum^i_{j=1} \mathbbm{1}_{\left\{\pi^r_j  \in \mathcal{U} \setminus \mathcal{N}\right\}}$, $\alpha^r_{i} = c(r,\pi^r_{i+1}) -  c(r,\pi^r_{i}) $ and $K^r = \min\{i \in \mathcal{U}: \sigma^r_i = k\}$ for every $(r,i) \in \mathcal{R} \times \mathcal{U}$.

\label{lemma:Ca_expansion}
\end{lemma}
\begin{proof}

Let $\tilde{c}(r,o), \forall (r,o) \in \mathcal{R} \times \mathcal{U}$ be a cost defined as $\tilde{c}(r, {\pi^r_i}) = c(r, {\pi^r_{i}}), \forall i \in [K^r]$ and 0 otherwise, and we also define $\tilde{\alpha}^r_i$  as $\tilde{\alpha}^r_i \triangleq \tilde{c}(r, {\pi^r_{i+1}}) - \tilde{c}(r, {\pi^r_i})$.

When $i = 1$ and for any $r \in \mathcal{R}$, we have
\begin{align}
      \tilde{c}(r, {\pi^r_{i}}) \min\left\{k, \sum^{i-1}_{j=1} x_{\pi^r_j}\right\} = 0.  \label{eq:null_term1}
\end{align}

Note that $\sigma^r_1 = 0$ by definition, since $\pi^r_1 \in \mathcal{N}$ for any $r \in \mathcal{R}$. We have
\begin{align}
   - \sum^{K^r}_{i=1} \tilde{\alpha}^r_i \sigma^r_i &=   \sum^{K^r}_{i=1}  \tilde{c}(r, {\pi^r_{i}}) \sigma^r_i - \sum^{K^r}_{i=1}  \tilde{c}(r, {\pi^r_{i+1}}) \sigma^r_i = \sum^{K^r}_{i=1}  \tilde{c}(r, {\pi^r_{i}}) \sigma^r_i - \sum^{K^r}_{i=2}  \tilde{c}(r, {\pi^r_{i}}) \sigma^r_{i-1}\\
   &= \sum^{K^r}_{i=1}  \tilde{c}(r, {\pi^r_{i}}) \sigma^r_i - \sum^{K^r}_{i=1}  \tilde{c}(r, {\pi^r_{i}}) \sigma^r_{i-1} =\sum^{K^r}_{i=1}  \tilde{c}(r, {\pi^r_{i}}) \mathds{1} _{\left\{ \pi^r_i \in \mathcal{U} \setminus \mathcal{N}\right\}} = \sum^{K^r}_{i=1}  c(r,\pi^r_{i}) \mathds{1}_{\left\{\pi^r_i \in \mathcal{U} \setminus \mathcal{N}\right\}}. \label{eq:sigma_sum}
\end{align}

Observe that the indicator function $\mathds{1} _{\left\{ \sum^{i-1}_{j=1} x_{\pi^r_j} < k \right\}}$ is 0 for every $i \geq K^r + 1$; therefore, the summation in Eq.~\eqref{eq:Ca_cost_function} can be limited to $K^r$ instead of $2N$. Using Lemma~\ref{lemma:mins_subtraction} we expand the expression of $C(r, \x)$ as follows: 
\begin{align}
    C(r, \x) &=  \sum^{K^r}_{i=1} c(r,\pi^r_{i}) x_{\pi^r_i} \mathds{1} _{\left\{ \sum^{i-1}_{j=1} x_{\pi^r_j} < k \right\}} =    \sum^{K^r}_{i=1} \tilde{c}(r, {\pi^r_{i}}) x_{\pi^r_i}\mathds{1}_{\left\{\sum^{i-1}_{j=1} x_{\pi^r_j} < k \right\} }\\
     &\stackrel{\eqref{eq:mins_subtraction}}{=}  \sum^{K^r}_{i=1}\tilde{c}(r, {\pi^r_{i}}) \left(\min\left\{k, \sum^i_{j=1} x_{\pi^r_j}\right\}-  \min\left\{k, \sum^{i-1}_{j=1} x_{\pi^r_j}\right\} \right) \\
      &\stackrel{\eqref{eq:null_term1}}{=} \sum^{K^r}_{i=1}  \tilde{c}(r, {\pi^r_{i}})  \min\left\{k, \sum^i_{j=1} x_{\pi^r_j}\right\} -  \sum^{K^r-1}_{i=1}  \tilde{c}(r, {\pi^r_{i+1}})  \min\left\{k, \sum^i_{j=1} x_{\pi^r_j}\right\} \\
      &=\sum^{K^r}_{i=1}  \tilde{c}(r, {\pi^r_{i}})  \min\left\{k, \sum^i_{j=1} x_{\pi^r_j}\right\} -  \sum^{K^r}_{i=1}  \tilde{c}(r, {\pi^r_{i+1}})  \min\left\{k, \sum^i_{j=1} x_{\pi^r_j}\right\} \\
      &=  - \sum^{ K^r}_{i=1} \tilde{\alpha}^r_i \min\left\{k , \sum^i_{j=1} x_{\pi^r_j}  \right\} =  - \sum^{ K^r}_{i=1} \tilde{\alpha}^r_i \min\left\{k - \sigma^r_i, \sum^i_{j=1} x_{\pi^r_j} - \sigma^r_i \right\} - \sum^{K^r}_{i=1}\tilde{\alpha}^r_i \sigma^r_i\\
    &\stackrel{\eqref{eq:sigma_sum}}
    {=} -  \sum^{ K^r}_{i=1} \tilde{\alpha}^r_i \min\left\{k - \sigma^r_i, \sum^i_{j=1} x_{\pi^r_j} - \sigma^r_i \right\} +  \sum^{K^r}_{i=1}  c(r,\pi^r_{i}) \mathds{1} _{\left\{ \pi^r_i \in \mathcal{U} \setminus \mathcal{N}\right\}}\\
    &= -  \sum^{ K^r-1}_{i=1} \alpha^r_i \min\left\{k - \sigma^r_i, \sum^i_{j=1} x_{\pi^r_j} - \sigma^r_i \right\} + \sum^{K^r}_{i=1}  c(r,\pi^r_{i}) \mathds{1} _{\left\{ \pi^r_i \in \mathcal{U} \setminus \mathcal{N}\right\}} \quad\quad  (\mathrm{since}~k- \sigma^r_{K^r} = 0).
\end{align}
This gives the cost function expression in Eq.~\eqref{eq:Ca_expansion_simple}.
\end{proof}

\begin{lemma}
\label{lemma:gain_expression}
For any request $r \in \mathcal{R}$ the caching gain function $G (r, \x)$ in Eq.~\eqref{e:gain}   has the following expression 
\begin{align}
 G(r, \x) =\sum^{K^r - 1
    }_{i=1} \alpha^r_i \min\left\{k - \sigma^r_i, \sum^i_{j=1} x_{\pi^r_j} - \sigma^r_i \right\}, 
\end{align}
where $\sigma^r_i =  \sum^i_{j=1} \mathbbm{1}_{\left\{\pi^r_j  \in \mathcal{U} \setminus \mathcal{N}\right\}}$, $\alpha^r_{i} = c(r,\pi^r_{i+1}) -  c(r,\pi^r_{i}) $ and $K^r = \min\{i \in \mathcal{U}: \sigma^r_i = k\}$ for every $(r,i) \in \mathcal{R} \times \mathcal{U}$.
\label{eq:appendix_gain_expression}
\end{lemma}
\begin{proof}
From Lemma~\ref{lemma:Ca_expansion}, the cost function $C(r, \x)$ can be equivalently expressed as:
\begin{align}
C(r,\x) &= -   \sum^{K^r -1
    }_{i=1} \alpha^r_i \min\left\{k - \sigma^r_i, \sum^i_{j= 1} x_{\pi^r_j} - \sigma^r_i \right\} +  \sum^{K^r}_{i=1}  c(r,\pi^r_{i}) \mathds{1}_{\left\{\pi^r_i \in \mathcal{U} \setminus \mathcal{N}\right\}}. 
\end{align} 
Without caching the system incurs the cost $k \Cret + \sum^k_{o \in \knn(r)} c(r, o )$ when $r \in \mathcal{R}$ is requested. This is the retrieval cost of fetching $k$ objects and the sum of the approximation costs of the $k$ closest objects in $\mathcal{N}$. From the definition of $\pi^r$, this cost can equivalently be expressed as $\sum^{K^r}_{i=1}  c(r,\pi^r_{i}) \mathds{1}_{\left\{\pi^r_i \in \mathcal{U} \setminus \mathcal{N}\right\}}$.  Therefore, we recover the gain expression in Eq.~\eqref{eq:appendix_gain_expression} as the cost reduction due to having a similarity cache. Thus, we have
\begin{align}
     G(r, \x) &=k \Cret + \sum^k_{o \in \knn(r)} c(r, o ) - C(r, \x)\\
     &= \sum^{K^r}_{i=1}  c(r,\pi^r_{i}) \mathds{1}_{\left\{ \pi^r_i \in \mathcal{U} \setminus \mathcal{N}\right\}} -  \sum^{K^r}_{i=1}  c(r,\pi^r_{i}) \mathds{1}_{\left\{ \pi^r_i \in \mathcal{U} \setminus \mathcal{N}\right\}} + \sum^{K^r -1
    }_{i=1} \alpha^r_i \min\left\{k - \sigma^r_i, \sum^i_{j= 1} x_{\pi^r_j} - \sigma^r_i \right\} \\
    &=  \sum^{K^r -1
    }_{i=1} \alpha^r_i \min\left\{k - \sigma^r_i, \sum^i_{j= 1} x_{\pi^r_j} - \sigma^r_i \right\}.
\end{align}
This concludes the proof.
\end{proof}
\section{Supporting Lemmas for Proof of Proposition~\ref{proposition:gain_upper_lower_bounds}}
\begin{lemma}
For every request $r \in \mathcal{R}$, index  $i \in \mathcal{U}$, and fractional cache state $\y \in \convX$ the index set defined as
\begin{align}
     I^r_i \triangleq \left\{ j \in [i]: \left(\pi^r_j \in \mathcal{N} \right)\land \left(\pi^r_j + N \notin \{\pi^r_l: l \in [i]\}\right)\right\}
     \label{eq:decisionsubset} 
\end{align}
satisfies the following
\begin{align}
    \sum_{j \in [i]} y_{\pi^r_j} -\sigma^r_i =  \sum_{j \in I^r_i} y_{\pi^r_j},
    \label{eq:decisionsubset_prop}
\end{align}
where $\sigma^r_i = \sum^i_{j=1} \mathds{1}_{\left\{\pi^r_j  \in \mathcal{U} \setminus \mathcal{N} \right\}}$ (defined in Eq.~\eqref{eq:definition_of_sigma}).
\end{lemma}
\begin{proof}
\begin{align}
    \sum_{j \in [i]} y_{\pi^r_j} -\sigma^r_i &\stackrel{\eqref{eq:definition_of_sigma}}{=}   \sum_{j \in [i]} y_{\pi^r_j} - \sum_{j\in[i]} \mathds{1}_{\{\pi^r_j \in \mathcal{U} \setminus \mathcal{N}\}} = \sum_{\substack{j \in [i] \\ \pi^r_j \in \mathcal{N}}} y_{\pi^r_j} + \sum_{\substack{j \in [i] \\ \pi^r_j \in \mathcal{U} \setminus \mathcal{N}}} \left(y_{\pi^r_j} - 1\right) = \sum_{\substack{j \in [i] \\ \pi^r_j \in \mathcal{N}}} y_{\pi^r_j}-  \sum_{\substack{l \in [i] \\ \pi^r_l \in \mathcal{U} \setminus \mathcal{N}}} y_{\left(\pi^r_l - N\right)}
    \label{eq:filtering_index1}
\end{align}
Remark that if  $l \in [i]$ and $\pi^r_l \in \mathcal{U} \setminus \mathcal{N}$, then the object $\pi^r_l - N$ has a strictly smaller cost and it appears earlier in the permutation $\pi^r$, that is there exists  $j < l$ such that $\pi^r_{j} = \pi^r_l - N$. In this case, the variable $y_{\left(\pi^r_l - N\right)}$ cancels out $y_{\pi^r_{j}}$ in the RHS of \eqref{eq:filtering_index1}. Then, we have   
\begin{align}
  \sum_{\substack{j \in [i] \\ \pi^r_j \in \mathcal{N}}} y_{\pi^r_j}-  \sum_{\substack{l \in [i] \\ \pi^r_l \in \mathcal{U} \setminus \mathcal{N}}} y_{\left(\pi^r_l - N\right)} = \sum_{\substack{j \in [i] \\ \pi^r_j \in \mathcal{N} \\ \pi^r_j + N \notin \left\{\pi^r_l: \left(l \in [i]\right)\land \left(\pi^r_l \in \mathcal{U} \setminus \mathcal{N} \right)\right\}}} y_{\pi^r_j}.
\end{align}
Note that if $\pi_l^r \notin \mathcal U \setminus \mathcal N$ (i.e., $\pi_l^r \in  \mathcal N$), then $\pi_j^r + N \neq \pi_l^r$. Then the sets $\left\{\pi^r_l: \left(l \in [i]\right)\land \left(\pi^r_l \in \mathcal{U} \setminus \mathcal{N} \right)\right\}$ and $\left\{\pi^r_l: l \in [i] \right\}$ coincide.
Therefore, the above equation can be simplified as follows 
\begin{align}
 \sum_{\substack{j \in [i] \\ \pi^r_j \in \mathcal{N}}} y_{\pi^r_j}-  \sum_{\substack{l \in [i] \\ \pi^r_l \in \mathcal{U} \setminus \mathcal{N}}} y_{\left(\pi^r_l - N\right)} 
 = \sum_{\substack{j \in [i] \\ \pi^r_j \in \mathcal{N} \\ \pi^r_j + N \notin \left\{\pi^r_l: l \in [i]\right\}}} y_{\pi^r_j}
 = \sum_{j \in I^r_i} y_{\pi^r_j}.\label{eq:filtering_index2}
\end{align}
Eq.~\eqref{eq:filtering_index1} and Eq.~\eqref{eq:filtering_index2} are combined to get 
\begin{align}
   \sum_{j \in [i]} y_{\pi^r_j} -\sigma^r_i=  \sum_{j \in I^r_i} y_{\pi^r_j},
\end{align}
and this concludes the proof.

\end{proof}

\section{Bounds on the Auxiliary Function}
We define $\El: \mathcal{R} \times \convX \to \mathbb{R}_+$, an auxiliary function,  that will be utilized in bounding the value of the gain function
\begin{align}
     \El(r,\y) &\triangleq   \sum^{K^r -1}_{i = 1 } \alpha^r_i   (k - \sigma^r_i)\left(1   -  \prod_{j \in I^r_i } \left(1 -\frac{y_{\pi^r_j}}{k - \sigma^r_i}\right)\right), \forall r \in \mathcal{R}, \y \in \convX \label{eq:useful_function_knn}.
\end{align}

The \DepRound{} \cite{byrka2014improved} subroutine outputs a rounded variable $\x \in \mathcal{X}$ from a fractional input $\y \in \convX$, by iteratively modifying the fractional input $\y$. At each iteration the subroutine \Simplify{} that is part of \DepRound{} is executed on two yet unrounded variables $y_i, y_j \in (0,1)$ with $i,j \in \mathcal{N}$, until all the variables are rounded in $\mathcal{O}(N)$ steps. Note that only $y_i, \forall i \in \mathcal{N}$ is rounded, since $x_i \in \mathcal{U} \setminus \mathcal{N}$ is determined directly from $x_{i-N}$. The random output  of \DepRound{} \cite{byrka2014improved} subroutine $\x \in \mathcal{X}$ given  the input $\y \in \convX$ has the following properties:
\begin{enumerate}[label=P\arabic*]
    \item $\mathbb{E}[x_{i}] = y_{i}, \forall i \in \mathcal{N}. $ \label{eq:condition1}
    \item $\sum_{i \in \mathcal{N}} x_{i}  = k. $\label{eq:condition2} 
    \item $\forall S \subset \mathcal{N}, \mathbb{E}\left[\prod_{i \in S} \left(1 - x_{i}\right)\right] \leq \prod_{i \in S} \left(1 - y_{i}\right).$ \label{eq:condition3}
\end{enumerate}

\begin{lemma}
The random output $\x \in \mathcal{X}$ of \DepRound{} given the fractional cache state input $\vec y \in \convX$, or the random output $\x \in \{0,1\}^\mathcal{U}$ of \Round{} given an integral cache state $\x'$, and fractional cache states $\y, \y' \in \convX$ with $\mathbb{E}[\x'] = \y'$, satisfy the following for any request $r \in \mathcal{R}$
\begin{align}
    \mathbb{E}\left[\El(r, \x)\right] \geq \El(r, \y).
    \label{eq:El_lowerbound}
\end{align}
\end{lemma}
\begin{proof}
 
\begin{align}
     \mathbb{E}\left[\El(r, \x)\right] & \stackrel{\eqref{eq:useful_function_knn}}{=} \mathbb{E} \left[  \sum^{K^r -1}_{i = 1 } \alpha^r_i   (k - \sigma^r_i)\left(1   -  \prod_{j \in I^r_i } \left(1 -\frac{x_{\pi^r_j}}{k - \sigma^r_i}\right)\right)\right] =  \sum^{K^r -1}_{i = 1 } \alpha^r_i   (k - \sigma^r_i)\left(1   -  \mathbb{E} \left[ \prod_{j \in I^r_i } \left(1 -\frac{x_{\pi^r_j}}{k - \sigma^r_i}\right)\right]\right)\\&
     \leq \sum^{K^r -1}_{i = 1 } \alpha^r_i   (k - \sigma^r_i)\left(1   -   \prod_{j \in I^r_i } \left(1 -\frac{y_{\pi^r_j}}{k - \sigma^r_i}\right)\right) = \El(r, \y).
\end{align}
The second equality is obtained using the linearity of the expectation operator. The inequality is obtained using~\cite[Lemma~E.10]{sisalem2021inference} with $S = I^r_i$, $c_m = \frac{1}{k - \sigma^r_i}$, as $\sigma^r_i < k$ for $i < K^r$ in the case when  $\x$ is the output of \DepRound, and in the case when  $\x$ is an output of \Round, the inequality holds with equality since every $x_i$ for $i \in \mathcal{N}$ is an independent random variable.   

\end{proof}

\section{Bounds on the Gain function}

\begin{proposition}
\label{proposition:gain_upper_lower_bounds}
The caching gain function $G (r, \x)$  defined in  Eq.~\eqref{e:gain} has the following lower and upper bound for any request $r \in \mathcal{R}$ and fractional cache state $\y \in \convX$:
\begin{align}
 \El(r,\y) &\leq  G(r,\y) \leq \left(1-\frac{1}{e}\right)^{-1} \El(r,\y). \label{eq:knn_gain_upper_bound}
\end{align}
\label{proposition:bounds_knn}
\end{proposition}
\begin{proof}
We have the following
\begin{align}
 G(r, \y) &=\sum^{K^r - 1
    }_{i=1} \alpha^r_i \min\left\{k - \sigma^r_i, \sum_{j \in [i]} y_{\pi^r_j} - \sigma^r_i \right\} \stackrel{\eqref{eq:decisionsubset_prop}}{=} \sum^{K^r - 1
    }_{i=1} \alpha^r_i \min\left\{k - \sigma^r_i, \sum^i_{j\in I^r_i} y_{\pi^r_j} \right\}\\
&\geq\sum^{K^r - 1
    }_{i=1}\alpha^r_i   (k - \sigma^r_i)\left(1   -  \prod_{j \in I^r_i}  \left(1 -\frac{y_{\pi^r_j}}{k - \sigma^r_i}\right)\right)\label{eq:ineq_lo}\\&= \El(r,\y),
\end{align}
and 
\begin{align}
G(r,\y) &= \sum^{K^r - 1
    }_{i=1} \alpha^r_i \min\left\{k - \sigma^r_i, \sum_{j \in I^r_i} y_{\pi^r_j} \right\} \\
    &\leq  \left(1 - \frac{1}{e}\right)^{-1} \sum^{K^r - 1
    }_{i=1} \alpha^r_i  (k - \sigma^r_i) \left(1   -  \prod_{j \in I^r_i}  \left(1 -\frac{y_{\pi^r_j}}{k - \sigma^r_i}\right)\right) \label{eq:ineq_up}\\&= \El(r, \y).
\end{align}
The inequalities in Eq.~\eqref{eq:ineq_lo} and Eq.~\eqref{eq:ineq_up} are obtained through~\cite[Lemma~E.7]{sisalem2021inference}, and~\cite[Lemma~E.8]{sisalem2021inference}, respectively, by setting $c = k - \sigma^r_i$, and $q_i = 1$ for every $i \in [K^r-1]$.
\end{proof}


\section{Subgradients Computation}
\label{sec:subgradientcomputation}
\begin{theorem}
\label{theorem:subgradient_expression}
For any time slot $t \in [T]$, the vectors $\vec{g}_t$ given by Eq.~\eqref{eq:subgradient_expression} are subgradients of the caching gain function $G(r,\y)$ for request $r_t \in \mathcal{R}$ at fractional cache state $\y_t \in \convX$. 
\begin{align}
    \vec{g}_t = \left[\left(c\left(r_t, \pi_{i^{r_t}_*+1}^{r_t}\right) - c\left(r_t, l\right)\right) \mathds{1}_{\left\{l^{r_t}_*\leq i^{r_t}_*\right\}}\right]_{l \in \mathcal{N}},
    \label{eq:subgradient_expression}
\end{align}
where $i^{r_t}_*\triangleq \max\left\{i \in [K^{r_t} - 1]:\left( \sum^i_{j=1} y_{t, \pi^{r}_j} \leq k\right) \land \left(l + N \notin \{\pi^{r_t}_v: v \in [i]\} \right)\right\}$, $l^{r_t}_* \triangleq \left(\pi^{r_t}\right)^{-1} (l), \forall l \in \mathcal{N}$, and $\left(\pi^{r_t}\right)^{-1}$ is the inverse permutation of $\pi^{r_t}$.
\end{theorem}
\begin{proof}
For any request $r \in \mathcal{R}$, the function $f^{(r,i)} (\y) \triangleq \min\left\{k - \sigma^r_i, \sum^i_{j=1} y_{\pi^r_j} - \sigma^r_i \right\}$ is a concave function, i.e., a minimum of two concave functions (a constant and an affine function). The subdifferential of the function at point $\y$, using Theorem~\cite[Theorem 8.2 ]{mordukhovich2017geometric} is given as
\begin{align}
    \partial f^{(r, i)} (\y)  = \begin{cases}
    \vec{0} & \mathrm{if}  \sum^i_{j=1} y_{\pi^r_j} > k, \\
    \mathrm{conv} \left(\left\{ \vec{0}, \nabla \left(\sum^i_{j=1} y_{\pi^r_j} \right) \right\}\right) & \mathrm{if}  \sum^i_{j=1} y_{\pi^r_j} = k, \\
    \nabla \left(\sum^i_{j=1} y_{\pi^r_j} \right) & \mathrm{otherwise}, 
    \end{cases} 
\end{align}
where $\mathrm{conv} \left(\,\cdot\,\right)$ is the convex hull of a set. Thus, a valid subgradient $\vec{g}^{(r,i)}(\y)$ of $f^{(r,i)}$ at point $\y$ can be picked as
\begin{align}
    \vec{g}^{(r,i)}(\y) = \begin{cases}
    \vec{0} & \mathrm{if}  \sum^i_{j=1} y_{\pi^r_j} \geq k, \\
    \nabla \left(\sum^i_{j=1} y_{\pi^r_j} \right)  & \mathrm{otherwise}.
    \end{cases}
\end{align}
Note that 
\begin{align}
    \frac{\partial}{\partial y_l} \sum^i_{j=1} y_{\pi^r_j} &= \mathds{1}_{\left\{y_l~\text{appears in the sum and}~y_{l+N}=1-y_l~\text{does not}\right\}}\\
    &= \mathds{1}_{\left\{\left(l \in \{\pi^r_v: v \in [i] \}\right) \,\,\land\,\, \left(l+N \not\in  \{\pi^r_v: v \in [i] \}\right)\right\}}\\
    &=  \mathds{1}_{\left\{\left(l^{r}_* \leq i\right) \,\,\land\,\, \left(l+N \not\in  \{\pi^r_v: v \in [i] \}\right)\right\}}.
\end{align}

The $l$-th component of the subgradient $\vec{g}^{(r,i)}(\y)$ is given by
\begin{align}
    {g}^{(r,i)}_l(\y) &= \begin{cases}
    0, & \mathrm{if}  \sum^i_{j=1} y_{\pi^r_j} \geq k, \\
      \frac{\partial}{\partial y_l} \sum^i_{j=1} y_{\pi^r_j} & \mathrm{otherwise}.\\
    \end{cases} \\ 
    &= \mathds{1}_{\left\{\left(\sum^i_{j=1} y_{\pi^r_j} < k \,\,\land\,\, l^{r}_* \leq i\right) \,\,\land\,\, \left(l+N \not\in  \{\pi^r_v: v \in [n] \}\right)\right\}}.
\end{align}

For any non-negative factor $\alpha^r_i$, we have $\partial \left(\alpha^r_i f^{(r, i)} (\y)\right) = \alpha^r_i  \partial \left( f^{(r, i)} (\y)\right) $ (multiply both sides of the subgradient inequality by a non-negative constant \cite[Definition 1.2]{bubeck2015convexbook}), and using \cite[Theorem 23.6]{rockafellar2015convex} we get
\begin{align}
    \partial G(r, \y) &= \partial \left(\sum^{K^r - 1
    }_{i=1} \alpha^{r_t}_i f^{(r, i)} (\y)\right)= \sum^{K^r - 1
    }_{i=1} \alpha^{r_t}_i \partial  f^{(r, i)} (\y).
\end{align}
Let $i^{r}_*\triangleq \max\{i \in [K^{r} - 1]: \left(\sum^i_{j=1} y_{\pi^{r}_j} \leq k \right) \land \left( l + N \notin \{\pi^{r}_v: v \in [i] \right)\} \}$. Now we can define a subgradient $\vec{g}_t$ of the function $G(r, \y)$ at point $\y_t \in \convX$ and request $r_t \in \mathcal{R}$ for any $t \in [T]$, where every component $l \in \mathcal{N}$ of $\vec{g}_t$ is given by
\begin{align}
    {g}_{t,l} &=\sum^{K^{r_t} - 1
    }_{i=1} \alpha^{r_t}_i {g}^{(r_t,i)}_l(\y_t) =\sum^{K^{r_t} - 1
    }_{i=1} \alpha^{r_t}_i  \mathds{1}_{\left\{\left(\sum^i_{j=1} y_{\pi^{r_t}_j} < k \right)\,\,\land\,\, \left( l^{r_t}_* \leq i\right) \,\,\land\,\, \left(l+N \not\in  \{\pi^{r_t}_v: v \in [n] \}\right)\right\}}\\
    &=\sum^{K^{r_t} - 1
    }_{i=l^{r_t}_*} \alpha^{r_t}_i  \mathds{1}_{\left\{\left(\sum^i_{j=1} y_{\pi^{r_t}_j} < k\right) \,\,\land\,\,\left( l+N \not\in  \{\pi^{r_t}_v: v \in [n] \}\right)\right\}} =\sum^{i^{r_t}_*
    }_{i=l^{r_t}_*} \alpha^{r_t}_i = \sum^{i^{r_t}_*
    }_{i=l^{r_t}_*} \left(c(r_t, \pi_{i+1}^{r_t}) - c(r_t, \pi_{i}^{r_t})\right)  \\
    &= \left(c(r_t, \pi_{i^{r_t}_*+1}^{r_t}) - c(r_t, l)\right) \mathds{1}_{\left\{l^{r_t}_*\leq i^{r_t}_*\right\}}, \forall{ l \in \mathcal{N}}.
\end{align}
Note that in the last equality we used the definition of $l^{r_t}_{*}$ to obtain $c(r_t, l) = c(r_t, \pi^r_{l^{r_t}_{*}})$. This concludes the proof.

\end{proof}
\section{Mirror Maps}
\label{sec:mirror_maps}
Let $\mathcal{D} \subset{\mathbb R^N}$ be a convex open set, and $\mathcal{Y}$ be a convex set such that $\mathcal{Y} \subset \mathrm{cl} (\mathcal{D})$ where $\mathrm{cl} (\,\cdot\,)$ is the closure; moreover, $\mathcal{D} \cap \mathcal{Y} \neq \emptyset$. The map $\Phi: \mathcal{D} \to \mathbb{R}$ is called a mirror map if the following is satisfied~\cite{bubeck2015convexbook}:
\begin{enumerate}
    \item The map $\Phi$ is strictly convex and differentiable. 
    \item The gradient of $\Phi$ takes all the possible values of $\mathbb R^N$, i.e., $\left\{\nabla \Phi(\x): \x \in \mathcal{D}\right\} = \mathbb R^N$.
    \item The gradient of $\Phi$ diverges at the boundary of $\mathcal{D}$, i.e., $\lim_{\x \to \partial \mathcal{D}} \norm{\nabla \Phi(\x)}  = +\infty$.
\end{enumerate}
It is easy to check the above properties are satisfied for the Euclidean mirror map $\Phi(\x) = \frac{1}{2} \norm{\x}^2$ over the domain $\mathcal{D} = \mathbb{R}^N$, and the negative entropy mirror map $\Phi(\x) = \sum_{i \in \mathcal{N}} x_i \log(x_i)$ over the domain $ \mathcal{D} = \mathbb{R}_{> 0}^N$.

\section{Supporting Lemmas for Proof of Theorem~\ref{theorem:knn_total_cost}}
\subsection{Subgradient Bound}
\label{sec:subgradient_constant}
\begin{lemma}
For any time slot $t \in [T]$, fractional cache state $\y_t \in \convX$, and request  $r_t \in \mathcal{R}$ the subgradients $\vec g_t$ of the gain function in Eq.~\eqref{eq:gain_compact} are bounded w.r.t the norm $\norm{\,\cdot\,}_\infty$ by the constant 
\label{lemma:subgradient_upper_bound}
\begin{align}
    L \triangleq c_d^k + c_f. 
    \label{eq:subgradient_upper_bound}
\end{align}
The constant $c_d^k$ is an upper bound on the dissimilarity cost of the k-th closest object for any request in $\mathcal{R}$, and $c_f$ is the retrieval cost.
\end{lemma}
\begin{proof}
For any time slot $t \in [T]$ we have 
\begin{align}
    \norm{\vec g_t}_\infty &= \max\left\{ g_{t,l}: l \in \mathcal{N}\right\}\stackrel{\eqref{eq:subgradient_expression}}{=} \max\left\{ c(r_t, \pi^{r_t}_{i^{r_t}_* + 1}) - c(r_t, l) : l \in \mathcal{N}\right\}\\
    &\leq  c(r_t, \pi^{r_t}_{K^{r_t}}) - c(r_t, \pi^{r_t}_1) \leq c(r_t, \pi^{r_t}_{K^{r_t}}) \\
    &= c_f + c(r_t, \pi^{r_t}_{K^{r_t}-N})  \leq  c_f + c(r_t, \pi^{r_t}_{K^{r_t}-N})  \leq c_f + c_d^k.
\end{align}

\end{proof}
Note that $L_2 \triangleq \norm{{\partial_{\y}} G(r, \y)}_2$ can be as high as $\sqrt{N} L$ and $N$ can be very large; moreover, the regret upper bound  is proportional to \New{$L_2$ instead of $L$} when the Euclidean map is used as a mirror map (see \cite[Theorem 4.2]{bubeck2015convexbook}). This justifies why it is preferable to work with the negative entropy instantiation of $\OMA$ rather than the classical Euclidean setting.

\subsection{Bregman Divergence Bound}
\begin{lemma}
Let  $\y^* = \underset{{\y \in \convX}}{{\arg\max}} \sum^T_{t=1} G(r, \y)$ and $\y_1 = \underset{\y  \in \convX \cap \mathcal{D}}{\arg\min} \,\Phi(\y)$, the value of the Bregman divergence $D_\Phi(\y_*, \y_1)$ associated with the negative entropy mirror map $\Phi$ is upper bounded by the constant \begin{align}
    D \triangleq  h \log\left(\frac{N}{h}\right).
    \label{eq:bregman_divergence_bound}
\end{align}
\label{lemma:bregman_divergence_bound}
\end{lemma}
\begin{proof}
It is easy to check that  $y_{1,i} =\frac{h}{N}, \forall i \in \mathcal{N}$ ($\vec y_1$ has maximum entropy); moreover, we have $\Phi(\y) \leq 0, \forall \y \in \convX$. The first order optimality condition~\cite[Proposition 1.3]{bubeck2015convexbook} gives $-{\nabla \Phi(\y_1)}^T ( \y - \y_1) \leq 0, \forall \y \in \convX$. We have
\begin{equation}
    D_\Phi(\y_*, \y_1) = \Phi(\y_*) - \Phi(\y_1) - \nabla \Phi(\y_1)^T ( \y_* - \y_1) \leq \Phi(\y_*) - \Phi(\y_1) \leq  -\Phi(\y_1) = h \log\left(\frac{N}{h}\right).
\end{equation}
\end{proof}
\section{Update Costs}
\subsection{Proof of Theorem~\ref{theorem:coupling_scheme}}
\label{proof:coupling_scheme}
\begin{proof}$ $\newline
\noindent\textbf{First part.} We show that $\mathbb E \left[\x_{t+1}\right] = \y_{t+1}$. Take $\vec \delta \triangleq \y_{t+1} - \y_t$.

If $\delta_i >0$ for $i \in \mathcal{N}$, then:
\begin{align}
    \mathbb E \left[ x_{t+1,i}\right] &=  \mathbb E \left[ x_{t+1,i} |  x_{t,i} = 1 \right] \mathbb{P} (x_{t,i} = 1) + \mathbb E \left[ x_{t+1,i} |  x_{t,i} = 0 \right]\mathbb{P} (x_{t,i} = 0)\\
    &= y_{t,i} + \left(\frac{{\delta_i}}{1 -y_{t,i}} + 0\right) (1 -y_{t,i}) = y_{t,i} + \delta_i.
\end{align}

If $\delta_i < 0$ for $i \in \mathcal{N}$, then:
\begin{align}
    \mathbb E \left[ x_{t+1,i}\right] &=  \mathbb E \left[ x_{t+1,i} |  x_{t,i} = 1 \right] \mathbb{P} (x_{t,i} = 1) + \mathbb E \left[ x_{t+1,i} |  x_{t,i} = 0 \right]  \mathbb{P} (x_{t,i} = 0)\\
    &=  \left(\frac{{y_{t,i} + \delta_i}}{{y_{t,i}}}\right)y_{t,i} + 0= y_{t,i} + \delta_i.
\end{align}

Otherwise, when $\delta_i = 0$ for $i \in \mathcal{N}$ we have $\mathbb E[x_{t+1,i}] = \mathbb E[x_{t,i}]=y_{t,i} = y_{t,i} + \delta_i $. Therefore we have for any $i \in \mathcal{N}$
\begin{align}
    \mathbb E[\x_{t+1}] =\y_t + \vec \delta =   \y_{t+1}.
\end{align}
\noindent\textbf{Second part.}
For any $i \in \mathcal{N}$, we can have two types of movements: If $\delta_i<0$, then given that $x_{t,i} = 1$, we evict with probability $\frac{-\delta_i}{y_{t,i}}$. If $\delta_i>0$, then given that $x_{t,i} = 0$, we retrieve a file  with probability $\frac{\delta_i}{1-y_{t,i}}$. Thus the expected movement is given by:
\begin{align}
\mathbb E \left[\norm{\x_{t+1} -\x_{t}}_{1}\right]&= \sum_{i \in \mathcal{N}} \mathbb E [|x_{t+1,i} - x_{t+1,i}|]\\
&=\sum_{i \in \mathcal{N}} \mathbb E \left[|x_{t+1,i} - x_{t,i}| \big|x_{t,i} = 0\right] \mathbb{P} (x_{t,i} = 0)  + \mathbb E \left[|x_{t+1,i} - x_{t,i}| \big|x_{t,i} = 1\right] \mathbb{P} (x_{t,i} = 1)\\
&= \sum_{i \in \mathcal{N}} \left(\frac{\delta_i}{1-y_{t,i}}\mathds{1}_{\left\{\delta_i >0\right\}}\cdot (1-y_{t,i}) + \frac{-\delta_i}{y_{t,i}}\mathds{1}_{\left\{\delta_i <0\right\}} \cdot y_{t,i}  \right) \\
& =\sum_{i \in \mathcal{N}} |\delta_i| =
\sum_{i \in \mathcal{N}} |y_{t+1,i} - y_{t,i}|=  \norm{\y_{t+1} - \y_t}_1.
\end{align}
\end{proof}
\subsection{Proof of Theorem~\ref{theorem:omd_movement_cost}}
\label{proof:omd_movement_cost}
\begin{proof}

The negative entropy mirror map $\Phi$ is $\rho=\frac{1}{h}$ strongly convex w.r.t the norm $\norm{\,\cdot\,}_1$ over $\mathcal{D} \cap \convX$ (see \cite[Ex.~2.5]{shalev2011online}), and the subgradients are bounded under the dual norm $\norm{\,\cdot\,}_\infty$ by $L$, i.e., for any $r_t \in \mathcal{R}$, $\y_t \in \convX$, and $t \in [T]$ we have $\norm{\vec{g}_t}_\infty \leq L$ (Lemma~\ref{lemma:subgradient_upper_bound}).  For any time slot $t \in [T-1]$, it holds:
\begin{align}
D_\Phi(\y_t, \z_{t+1}) &= \Phi(\y_t) - \Phi(\z_{t+1}) - {\nabla \Phi(\z_{t+1})}^T (\y_t - \z_{t+1})  \nonumber\\
&=   \Phi(\y_t) - \Phi(\z_{t+1}) + {\nabla \Phi(\y_{t})}^T (\z_{t+1} - \y_t) \nonumber + {\left(\nabla \Phi(\y_{t}) - \nabla \Phi(\z_{t+1})\right)}^T (\y_t - \z_{t+1} )  \nonumber\\
&\leq - \frac{\rho}{2} \norm{\y_t - \z_{t+1}}_1^2 + \eta \vec g^T_t (\y_t - \z_{t+1})  \label{eq:strong_convx_update_rule}\\
&\leq   - \frac{\rho}{2} \norm{\y_t - \z_{t+1}}^2 + \eta L \norm{\y_t - \z_{t+1}}_1^2 \label{eq:cauchy_step}\\
&\leq \frac{\eta^2 L^2}{2 \rho} \label{eq:last_step}.
\end{align}
Eqs.~\eqref{eq:strong_convx_update_rule}--\eqref{eq:last_step} are obtained using the strong convexity of $\Phi$ and the update rule, Cauchy-Schwarz inequality, and the inequality $a x - b x^2 \leq \max_{x} a x - b x^2 = a^2/4b $ as in the last step in the proof of \cite[Theorem 4.2]{bubeck2015convexbook}, respectively. \Old{We obtain:}
\Old{\begin{align}
    \sum^{T-1}_{t=1} \norm{\y_{t+1} - \y_{t}}_1 &\leq \sum^{T-1}_{t=1}  \sqrt{\frac{2}{\rho} D_\Phi(\y_t, \y_{t+1}) }  \label{eq:strong_conv}\\
&\leq   \sum^{T-1}_{t=1} \sqrt{\frac{2}{\rho} D_\Phi(\y_t, \z_{t+1}) - \frac{2}{\rho} D_\Phi(\y_{t+1}, \z_{t+1})} \label{eq:pyth} \\
& \leq   \sum^{T-1}_{t=1} \sqrt{\frac{2}{\rho} D_\Phi(\y_t, \z_{t+1}) } \label{eq:conv_remove} \\
& \leq  \sum^{T-1}_{t=1} \sqrt{2 \eta^2 \frac{L^2}{2\rho^2}} \label{eq:const_reduc} \leq \frac{L\eta}{\rho}  T.
\end{align}}
\New{Moreover, for any $t \in [T-1]$ it holds}
\New{\begin{align}
     \norm{\y_{t+1} - \y_{t}}_1 &\leq   \sqrt{\frac{2}{\rho} D_\Phi(\y_t, \y_{t+1}) }
\leq    \sqrt{\frac{2}{\rho} D_\Phi(\y_t, \z_{t+1}) - \frac{2}{\rho} D_\Phi(\y_{t+1}, \z_{t+1})}\leq    \sqrt{\frac{2}{\rho} D_\Phi(\y_t, \z_{t+1}) }  \leq   \sqrt{2 \eta^2 \frac{L^2}{2\rho^2}}\leq \frac{L\eta}{\rho}. \label{eq:mc_chain}
\end{align}}
\Old{Eqs.~\eqref{eq:strong_conv}--\eqref{eq:const_reduc} are obtained using}\New{The above chain of inequalities is obtained through}: the strong convexity of $\Phi$, the generalized Pythagorean inequality \cite[ Lemma 4.1]{bubeck2015convexbook}, non-negativity of the Bregman divergence of a convex function, and  Eq.~\eqref{eq:last_step}, in respective order. The learning rate is  $\eta = \mathcal{O}\left(\frac{1}{\sqrt{T}}\right)$; therefore, we finally get
\begin{align}
    \sum^{T-1}_{t=1} \norm{\y_{t+1} - \y_{t}}_1\New{ \stackrel{\eqref{eq:mc_chain}}{\leq} \frac{L\eta T}{\rho}} = \mathcal{O}(\sqrt{T}).
\end{align}
\end{proof}

\section{Proof of Theorem~\ref{theorem:knn_total_cost}}
\label{proof:knn_total_cost}
\begin{proof} To prove the $\psi$-regret guarantee: (i) we first  establish an upper bound on the regret of the \acai{} policy over its fractional cache states domain $\convX$ against a fractional optimum, then (ii) the guarantee is transformed in a $\psi$-regret guarantee over the integral cache states domain $\mathcal{X}$ in expectation.

\noindent\textbf{Fractional domain guarantee.} We establish first the regret of running Algorithm~\ref{algo:online_mirror_ascent} with decisions taken over the fractional domain $\convX$. The following properties are satisfied:
\begin{enumerate}[label=(\roman*)]
    \item The caching gain function $G(r, \y)$ is concave over its fractional domain $\convX$ for any $r \in \mathcal{R}$ (see Sec.~\ref{sec:cache_state_service_gain}).
    \item The negative entropy mirror map $\Phi: \mathcal{D} \to \mathcal{R}$ is $\frac{1}{h}$ strongly convex w.r.t the norm $\norm{\,\cdot\,}_1$ over $\mathcal{D} \cap \convX$ (see \cite[Ex. 2.5]{shalev2011online}).
    \item The subgradients are bounded under the dual norm $\norm{\,\cdot\,}_\infty$ by $L$, i.e., for any $r_t \in \mathcal{R}$, $\y_t \in \convX$, and $t \in [T]$ we have $\norm{\vec{g}_t}_\infty \leq L$ (Lemma~\ref{lemma:subgradient_upper_bound}).
    \item The Bregman divergence $D_\Phi(\y^*, \y_1)$ is bounded by a constant $D$ where $\y^* = {\arg\max}_{\y \in \convX} \sum^T_{t=1} G(r, \y)$ and $\y_1 = \underset{\y  \in \convX \cap \mathcal{D}}{\arg\min} \,\Phi(\y)$ is the initial fractional cache state (Lemma~\ref{lemma:bregman_divergence_bound}).
    \end{enumerate}
With the above properties satisfied, the regret of Algorithm~\ref{algo:online_mirror_ascent} with the gains evaluated over the fractional cache states $\{\y_t\}^T_{t=1} \in \convX^T$ is \cite[Theorem 4.2]{bubeck2015convexbook} 
\begin{align}
   \Regret\convX{\OMA_{\Phi}} &=  \underset{\{\vec r_t\}^T_{t=1} \in \mathcal{R}^T}{\sup}\left\{ \sum^T_{t=1} G(r_t,\y_*)  - \sum^T_{t=1} G(r_t,\y_t) \right\} \\
   &\leq \frac{D_\Phi(\y_*, \y_1)}{\eta} + \frac{\eta h}{2} \sum^T_{t=1} \norm{\vec{g}_t}^2_{\infty} \leq \frac{D}{\eta} + \frac{\eta L^2 h T}{2}.
    \label{eq:regret_knn}
\end{align}
\Old{By selecting the learning rate $\eta=\frac{1}{L}\sqrt{\frac{2 D}{hT}}$ giving the tightest upper bound we obtain 
\begin{align}
    \mathrm{Regret}_{T, \convX} \leq L \sqrt{2 h D T} \stackrel{\eqref{eq:subgradient_upper_bound},\eqref{eq:bregman_divergence_bound}}{=} (c_d^k + c_f) h \sqrt{2 \log{\left(\frac{N}{h}\right)} T}.
\end{align}}
\noindent \textbf{Integral domain guarantee.} Let $\x_* =  \underset{\x \in \mathcal{X}}{\arg\max} \sum^T_{t=1} G(r_t,\x)$ and $\y_* =   \underset{\y \in \convX}{\arg\max} \sum^T_{t=1} G(r_t,\y)$. The fractional cache state $\y_* \in \convX$ is obtained by maximizing  $\sum^T_{t=1} G(r_t,\y)$ over the domain $\y \in \convX$. We can only obtain a lower gain by restricting the maximization to a subset of the domain $\mathcal{X} \subset \convX$. Therefore, we obtain:
\begin{align}
     \sum^T_{t=1} G(r_t,\y_*)  \geq \sum^T_{t=1} G(r_t,\x_*). \label{eq:optimality_in_convx}
\end{align}
 Note that the components for $\x$ and $\y$ in  $\mathcal{U} \setminus \mathcal{N}$ are completely determined by the components in $\mathcal{N}$. Take $\psi = 1 - 1/e$. \Old{Take $\psi = 1 - 1/e$, we get:
\begin{align}
    \mathbb{E} \left[\sum^T_{t=1} G (r_t, \x_t)\right] &\stackrel{\eqref{eq:knn_gain_upper_bound}}{\geq} \mathbb{E} \left[\sum^T_{t=1} \El(r_t, \x_t)\right] \stackrel{\eqref{eq:El_lowerbound}}{\geq} \sum^T_{t=1} \El(r_t, \y_t) \stackrel{\eqref{eq:knn_gain_upper_bound}}{\geq}  \psi \sum^T_{t=1} G (r_t, \y_t)  \nonumber \\&\stackrel{\eqref{eq:regret_knn}}{\geq}\psi \sum^T_{t=1} G (r_t, \y_*) -{A\sqrt{T}} 
     \stackrel{\eqref{eq:optimality_in_convx}}{\geq} \psi \sum^T_{t=1} G (r_t, \x_*) -{A\sqrt{T}}, 
\end{align}
where $A = \psi (c_d^k + c_f) h \sqrt{2 \log{\left(\frac{N}{h}\right)} }$.} \New{For every  $t \in \{1,M, 2M, \dots,  M \floor{T/M}\}$ and $r \in \mathcal{R}$ it holds
\begin{align}
    \mathbb{E} \left[G (r, \x_{t})\right] &\stackrel{\eqref{eq:knn_gain_upper_bound}}{\geq} \mathbb{E} \left[ \El(r, \x_t)\right] \stackrel{\eqref{eq:El_lowerbound}}{\geq}  \El(r, \y_t) \stackrel{\eqref{eq:knn_gain_upper_bound}}{\geq}  \psi  G (r, \y_t).\label{eq:sandwitch}
\end{align}
Moreover, consider the following decomposition of the time slots $\{1,2,\dots, T\} = \mathcal{T}_1\cup \mathcal{T}_2\cup\dots\cup\mathcal{T}_{\floor{T/M}+1}$, where each $\mathcal{T}_i$ represents the $i$-th freezing period, i.e., $\vec x_t = \x_{\min(\mathcal T_i)}$ for $t \in \mathcal{T}_i$.  Note that $\card{\mathcal{T}_i} \leq M$ for $i \in \{1,2, \dots, \floor{T/M} +1 \}$. Now we can decompose the total expected gain of the policy as 
\begin{align}
  \psi \sum^T_{t=1} G (r_t, \x_{*})- \sum^T_{t=1} \mathbb{E} \left[G (r_t, \x_{t})\right] &\stackrel{\eqref{eq:optimality_in_convx}}{\leq}  \psi  \sum^T_{t=1} G (r_t, \y_{*})- \sum^T_{t=1} \mathbb{E} \left[G (r_t, \x_{t})\right] \nonumber=  \psi  \sum^T_{t=1} G (r_t, \y_{*}) -  \sum^{\floor{T/M}+1}_{i=1}\sum_{t \in \mathcal{T}_i} \mathbb{E} \left[G (r_t, \x_{t})\right] \nonumber\\&\stackrel{\eqref{eq:sandwitch}}{\leq} \psi  \sum^T_{t=1} G (r_t, \y_{*})  -  \psi  \sum^{\floor{T/M}+1}_{i=1}\sum_{t \in \mathcal{T}_i} G (r_t, \y_{\min(\mathcal{T}_i)}) \nonumber\\&=  \psi  \sum^T_{t=1} G (r_t, \y_{*}) -\psi  \sum^T_{t=1} G (r_t, \y_{t}) + \psi  \sum^T_{t=1} G (r_t, \y_{t})  -  \psi  \sum^{\floor{T/M}+1}_{i=1}\sum_{t \in \mathcal{T}_i} G (r_t, \y_{\min(\mathcal{T}_i)})
  \nonumber\\ &\stackrel{\eqref{eq:regret_knn}}{\leq}\psi \cdot  \Regret\convX{\OMA_{\Phi}}  +  \psi \left(\sum^T_{t=1} G (r_t, \y_{t}) -  \sum^{\floor{T/M}+1}_{i=1}\sum_{t \in \mathcal{T}_i} G (r_t, \y_{\min(\mathcal{T}_i)})\right).\label{eq:rhs_freeze}
\end{align}
The first equality is obtained through a decomposition of the time slots to $\mathcal{T}_i$ for $ i \in \{1, 2, \dots, \floor{T/M}+1\}$.  It remains to bound the r.h.s of Eq.~\eqref{eq:rhs_freeze}, i.e., 
\begin{align}
&\sum^T_{t=1} G (r_t, \y_{t}) -  \sum^{\floor{T/M}+1}_{i=1}\sum_{t \in \mathcal{T}_i} G (r_t, \y_{\min(\mathcal{T}_i)}) =\sum^{\floor{T/M}+1}_{i=1} \sum_{t \in \mathcal{T}_i} G (r_t, \y_{t}) - G (r_t, \y_{\min(\mathcal{T}_i)}) \\
&\leq \sum^{\floor{T/M}+1}_{i=1} \sum_{t \in \mathcal{T}_i} \partial_{\vec y} G (r_t, \y_{\min(\mathcal{T}_i)}) \cdot \left(\y_{t} -\y_{\min(\mathcal{T}_i)} \right) \qquad\qquad \text{concavity of $G(r_t, \cdot)$} \nonumber\\
&\leq\sum^{\floor{T/M}+1}_{i=1} \sum_{t \in \mathcal{T}_i} \norm{ \partial_{\vec y} G (r_t, \y_{\min(\mathcal{T}_i)}) }_\infty \norm{\y_{t} -\y_{\min(\mathcal{T}_i)} }_1 \qquad\qquad\text{H\"older's inequality} \nonumber\\
&\leq  L \sum^{\floor{T/M}+1}_{i=1} \sum_{t \in \mathcal{T}_i} \sum^{t-1}_{t'=\min(\mathcal{T}_i)} \norm{\y_{t+1} -\y_{t} }_1 \quad\quad\quad\quad\quad\quad\quad\text{triangle inequality and definition of $L$} \nonumber\\ 
&\leq  \frac{L^2}{\rho} \sum^{\floor{T/M}+1}_{i=1} \sum_{t \in \mathcal{T}_i} \sum^{t-1}_{t'=\min(\mathcal{T}_i)} \eta \leq  \frac{L^2 \eta}{\rho} \cdot \left(\frac{T}{M}+1\right) \cdot \frac{M (M-1)}{2}\quad\quad\quad\quad\quad\text{update cost upper bound in Eq.~\eqref{eq:mc_chain}} \nonumber\\
&= \frac{L^2 \eta}{2\rho} (M-1)\left(T+M\right)= \frac{h L^2 \eta}{2} (M-1)\left(T+M\right) \quad\quad\quad\quad\text{we have $\rho = 1/h$}.\label{eq:proof_integral_p2} &
\end{align}}
\New{Thus, by bounding r.h.s of Eq.~\eqref{eq:rhs_freeze} in Eq.~\eqref{eq:proof_integral_p2} we get
\begin{align}
    \psi \sum^T_{t=1} G (r_t, \x_{*})- \sum^T_{t=1} \mathbb{E} \left[G (r_t, \x_{t})\right] &\leq  \psi \left(\frac{D}{\eta} + \frac{\eta L^2 h}{2} T + \frac{\eta L^2 h}{2} (M-1)\left(T+M\right)\right)
\end{align}
By selecting the learning rate $\eta = \frac{1}{L}\sqrt{\frac{2 D}{h (T + (M-1) (M+T))}} = \frac{1}{(c_d^k + c_f)} \sqrt{\frac{2 \log\left(\frac{N}{h}\right)}{(T + (M-1) (M+T))}}$  giving the tightest upper bound we obtain 
\begin{align}
    \psi \sum^T_{t=1} G (r_t, \x_{*})- \sum^T_{t=1} \mathbb{E} \left[G (r_t, \x_{t})\right] &\leq \psi L \sqrt{ 2 D h (T + (M-1) (T+M))} \\
    &\stackrel{\eqref{eq:subgradient_upper_bound},\eqref{eq:bregman_divergence_bound}}{=} \left(1 - \frac{1}{e}\right)(c_d^k + c_f) h \sqrt{2 \log{\left(\frac{N}{h}\right)}  (T + (M-1) (T+M))}.
\end{align}
This concludes the proof.
}

\end{proof}
\section{Proof of Corollary~\ref{corollary:offline}}
\label{proof:offline}
We have the following 
\begin{align}
    \mathbb{E}\left[G_T(\bar\x)\right]  &\stackrel{\eqref{eq:knn_gain_upper_bound}}{\geq} \mathbb{E} \left[\frac{1}{T} \sum^T_{t=1} \El(r_t, \x_t)\right] \stackrel{\eqref{eq:El_lowerbound}}{\geq}\frac{1}{T} \sum^T_{t=1} \El(r_t, \x_t) \stackrel{\eqref{eq:knn_gain_upper_bound}}{\geq} \psi G_T(\bar y).
    \label{e:static_head}
\end{align}
We apply Jensen's inequality to obtain
\begin{align}
    G_T(\bar y) \geq \frac{1}{\tilde T} \sum^{\tilde T}_{i=1} G_T(\y_i).\label{e:static_jensen}
\end{align}
It is easy to verify that $G_T$~\eqref{eq:static_gain} is concave, and has bounded subgradients under the $l_\infty$ norm over the fractional caching domain $\convX$; moreover, the remaining properties are satisfied for the same mirror map and decision set. The regret of Algorithm~\ref{algo:online_mirror_ascent} with the gains evaluated over the fractional cache states $\{\y_i\}^{\tilde T}_{i=1} \in \convX^T$ is \cite[Theorem 4.2]{bubeck2015convexbook} is given by
\begin{align}
    \sum^{\tilde T}_{i=1} G_T(\y_*) - \sum^{\tilde T}_{i=1} G_T(\y_i) = {\tilde T} G_T(\y_*) - \sum^{\tilde T}_{i=1} G_T(\y_i)  \stackrel{\eqref{eq:regret_knn}}{\leq} (c_d^k + c_f) h \sqrt{2 \log{\left(\frac{N}{h}\right)} {\tilde T}}.
\end{align}
Divide both sides of the equality by ${\tilde T}$, and move the gain attained by \acai{} to the l.h.s to get 
\begin{align}
  \frac{1}{{\tilde T}} \sum^{\tilde T}_{i=1} G_T(\y_i)  \geq G_T(\y_*)  - (c_d^k + c_f) h \sqrt{2 \log{\left(\frac{N}{h}\right)} {\tilde T}}\label{e:static_regret},
\end{align}
where $\y_* = \underset{\y \in \convX}{\argmax} G_T(\y)$.

We combine Eq.~\eqref{e:static_head}, Eq.~\eqref{e:static_jensen}, and Eq.~\eqref{e:static_regret} to obtain 
\begin{align}
     \mathbb{E}\left[G_T(\bar\x)\right] \geq G_T(\y_*)  - (c_d^k + c_f) h \sqrt{2 \log{\left(\frac{N}{h}\right)} {\tilde T}},
\end{align}
and $G_T(\y_*)$ can only be larger than $G_T(\x_*)$; thus, we also obtain
\begin{align}
     \mathbb{E}\left[G_T(\bar\x)\right] \geq G_T(\x_*)  - (c_d^k + c_f) h \sqrt{\frac{2 \log{\left(\frac{N}{h}\right)}}{{\tilde T}}}.
\end{align}
We conclude $\forall \epsilon >0$ for a sufficiently large number of iterations ${\tilde T}$, $\bar \x$ satisfies
\begin{align}
\mathbb{E} \left[ G_T(\bar{\x})\right]\geq \left(1-\frac{1}{e} - \epsilon\right) G_T({\x_*}).
\end{align}

\section{Additional Experiments}

\subsection{Redundancy}
We quantify the redundancy present in the caches in Figure~\ref{fig:amazon_storage_redundacy} (a), as the percentage of added objects to fill the physical cache. We also show the contribution of the dangling objects to the gain in Figure~\ref{fig:amazon_storage_redundacy} (b), that does not exceed $2.0\%$ under both traces. 
\begin{figure}[h]
\centering
    \subcaptionbox{Storage redundancy of the \lru-based policies under \emph{SIFT1M trace}.}{
    \includegraphics[width=.3\linewidth]{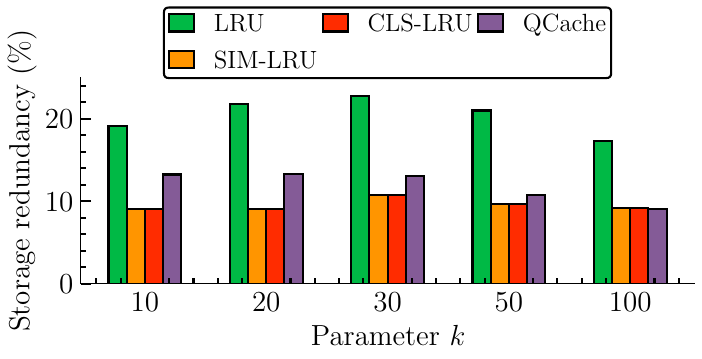}
  } 
  \subcaptionbox{Storage redundancy of the \lru-based policies under \emph{Amazon trace}.}{
    \includegraphics[width=.3\linewidth]{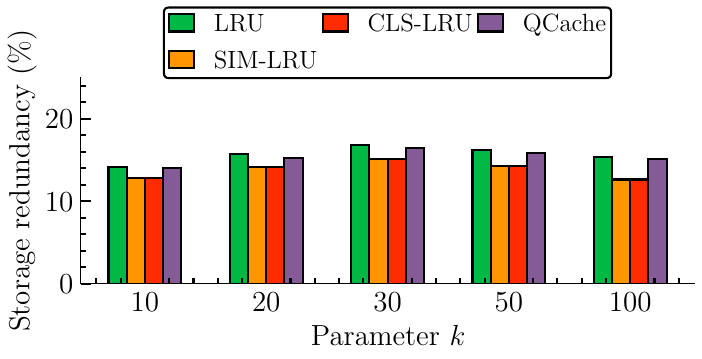}
  } \\
    \subcaptionbox{The average gain contribution of the dangling objects under {SIFT1M trace}.}{
    \includegraphics[width=.3\linewidth]{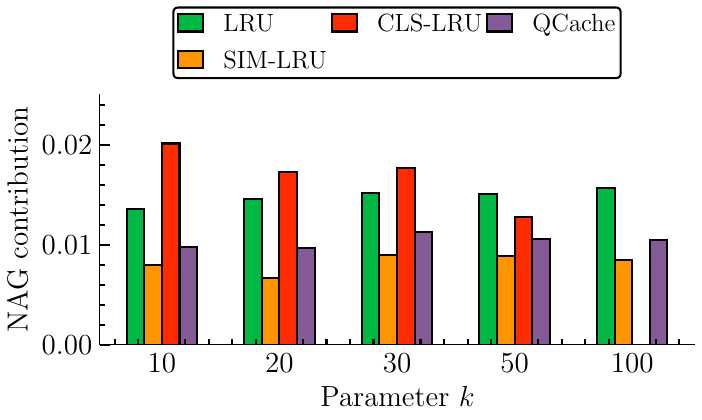}}
    \subcaptionbox{The average gain contribution of the dangling objects under {Amazon trace}.}{
    \includegraphics[width=.3\linewidth]{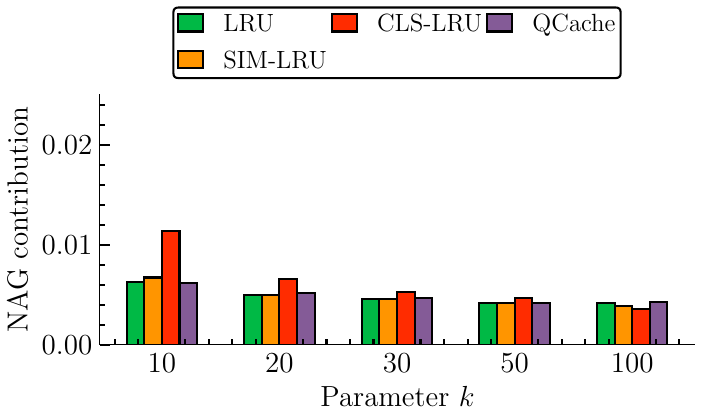}
  
  } 
  \caption{Storage redundancy percentage and gain contribution for the different policies. The cache size $h = 1000$, and $k \in \{10,20,30,50,100\}$. }
\label{fig:amazon_storage_redundacy}
\end{figure}
\subsection{Approximate Index Augmentation} We repeat the sensitivity analysis and show the caching gain when the different policies are augmented with an approximate index, and are allowed to mix the best object that can be served locally and from the server. Figure~\ref{fig:additional_amazon_gain_vs_cache_size} shows the caching gain for the different caching policies and different values of the cache size $h \in \{50, 100, 200, 500, 1000\}$ and $k= 10$. Figure~\ref{fig:additional_amazon_gain_vs_retrieval_cost} shows the caching gain for the different caching policies and different values of the retrieval cost $c_f$, that is taken as the average distance to the $i$-th neighbor, $i \in \{2, 50, 100, 500, 1000\}$. The cache size is $h= 1000$ and $k = 10$. Figure~\ref{fig:additional_amazon_gains_complete_knowledge} shows the caching gain for the different caching policies and different values of $k \in \{10,20,30, 50, 100\}$ and $h=1000$.

\begin{figure}[h]
    \centering
    \subcaptionbox{{SIFT1M trace}}{\includegraphics[width=.3\linewidth]{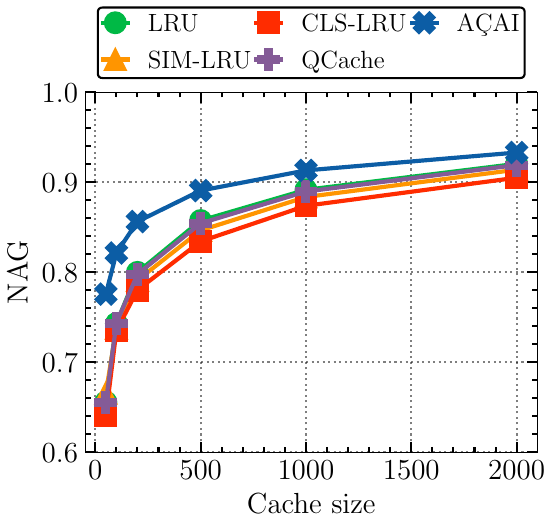}}
    \subcaptionbox{{Amazon trace}}{\includegraphics[width=.3\linewidth]{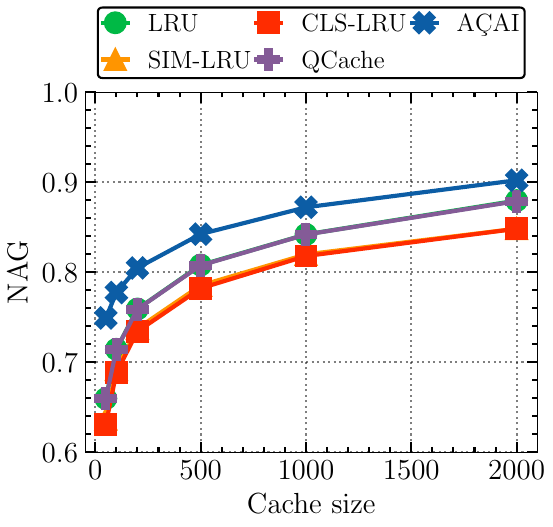}}

    \caption{Caching gain of the different policies when augmented with the approximate index, for  different cache sizes $h \in \{50,100,200,500,1000\}$ and $k= 10$. }
    \label{fig:additional_amazon_gain_vs_cache_size}
\end{figure}
\begin{figure}[h]
    \centering
        \subcaptionbox{SIFT1M trace}{
    \includegraphics[width=.3\linewidth]{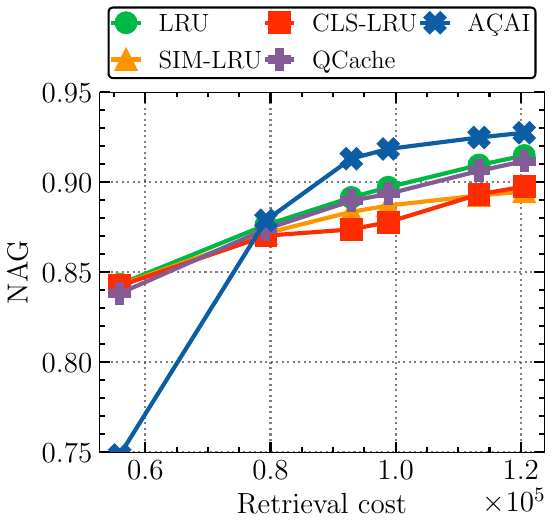}
    }
    \subcaptionbox{Amazon trace}{
    \includegraphics[width=.3\linewidth]{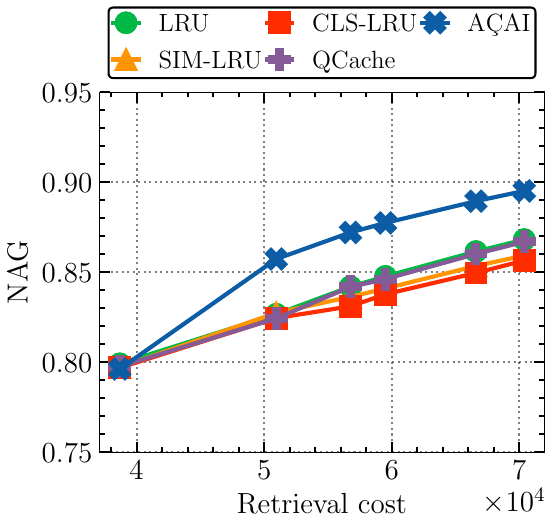}
    }
    \caption{Caching gain for the different policies and  different retrieval  cost when augmented with the approximate index. The retrieval cost $c_f$ is taken as the average distance to the $i$-th neighbor, $i \in \{2, 10,  50, 100, 500, 1000\}$. The cache size is $h= 1000$ and $k = 10$. }
    \label{fig:additional_amazon_gain_vs_retrieval_cost}
\end{figure}
\begin{figure}[h]
\centering
  \subcaptionbox{SIFT1M trace}{
    \includegraphics[width=.3\linewidth]{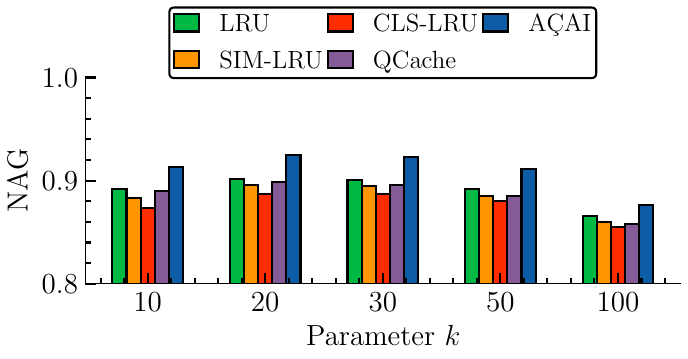}
  }
    \subcaptionbox{{Amazon trace}}{
    \includegraphics[width=.3\linewidth]{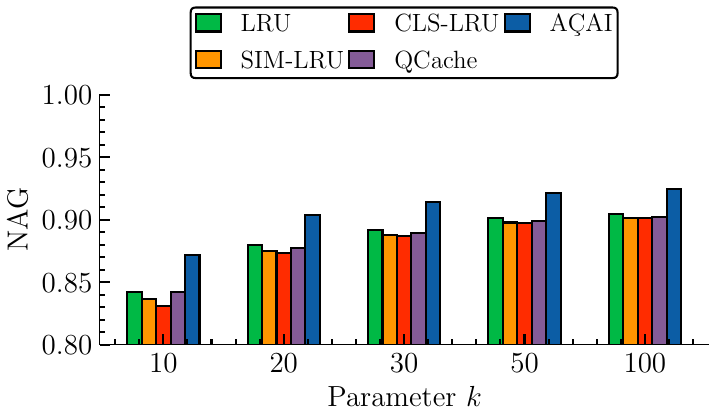}
  }
  \caption{Caching gain for the different policies  when augmented with the approximate index. The cache size is $h = 1000$, and $k \in \{10,20,30,50,100\}$.}
    \label{fig:additional_amazon_gains_complete_knowledge}
\end{figure}

\subsection{Computation Time}

\New{We provide a comparison of the computation time of the different algorithms in Fig.~\ref{fig:tc}.  When the different LRU-like policies are not augmented with a global catalog index in Fig.~\ref{fig:tc}~(a), \acai{} experiences a higher computation time per iteration. When the different LRU-like policies are augmented with a global catalog index in Fig.~\ref{fig:tc}~(b), \acai{} has a similar computation time to the different policies except for the simple vanilla lru policy. Nonetheless, in both settings \acai{} computation time remains comparable and approximately within factor $4$ w.r.t the computation time of the different policies. Due to space limitation we cannot add this figure to the main text.}

\begin{figure}[t]
    \centering
   \subcaptionbox{}{\includegraphics[width=.3\linewidth]{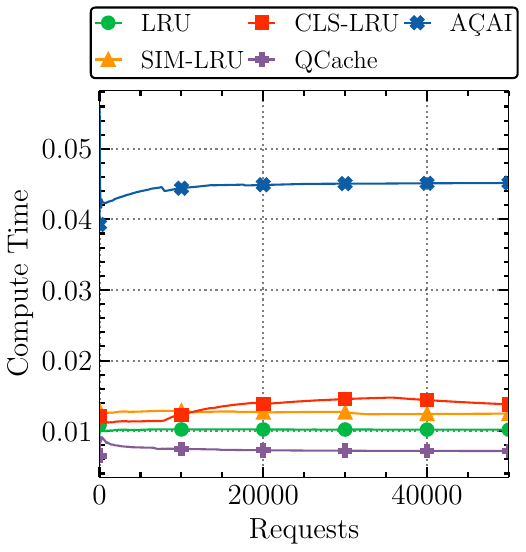}}
         \subcaptionbox{}{
    \includegraphics[width=.3\linewidth]{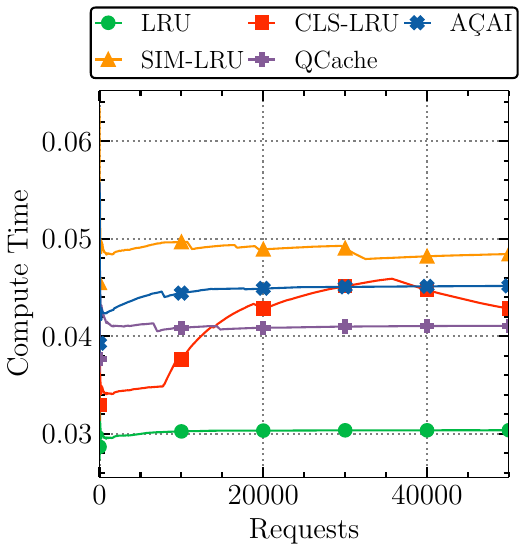}}
    \caption{\New{Time-averaged computation time of \acai{} and the different policies \lru, \simlru, \clslru, and \qcache (a)  w/o an approximate global catalog index, and (b) w/ an approximate global catalog index. Experiment run over the Amazon trace, cache size $h = 1000$, parameter $k = 10$, $\eta = 10^{-4}$, retrieval cost $c_f$ is set to be the distance to the $50$-th closest neighbor in $\mathcal{N}$.}}
    \label{fig:tc}
\end{figure}

\fi
\newpage

\end{document}